\documentclass[ba]{imsart}
\pubyear{2024}
\volume{TBA}
\issue{TBA}
\arxiv{2110.13017}
\firstpage{1}
\lastpage{1}

\usepackage[table]{xcolor}
\usepackage{amsthm}
\usepackage{amsmath}
\usepackage{natbib}
\usepackage[colorlinks,citecolor=blue,urlcolor=blue,filecolor=blue,backref=page]{hyperref}
\usepackage{graphicx}
\usepackage{tikz}
\usepackage{amsmath}
\usepackage{amsfonts}
\usepackage{cleveref}
\usepackage{amssymb}

\usepackage{amsthm}
\newtheorem{theorem}{Theorem}[section]
\newtheorem{definition}[theorem]{Definition}
\newtheorem{lemma}[theorem]{Lemma}
\newtheorem{corollary}[theorem]{Corollary}

\newtheorem{remark}[theorem]{Remark}

\newcommand{\dd}{\mathrm{d}}

\definecolor{charlesBlue}{RGB}{100, 155, 255}
\definecolor{lightgray}{RGB}{200, 200, 200}

\startlocaldefs
\newcommand{\nR}{\widehat{R}_{\nu}}
\newcommand{\nB}{\widehat{B}_\nu}
\newcommand{\nW}{\widehat{W}_\nu}
\newcommand{\B}{B_\nu}
\newcommand{\W}{W_\nu}

\endlocaldefs

\begin{document}


\begin{frontmatter}
\title{Nested $\widehat R$: Assessing the convergence of Markov chain Monte Carlo when running many short chains}
\runtitle{Nested $\widehat R$}

\begin{aug}
\author[A]{\fnms{Charles C.}~\snm{Margossian}\ead[label=e1]{cmargossian@flatironinstitute.org}},
\author[B]{\fnms{Matthew D.}~\snm{Hoffman}},
\author[B]{\fnms{Pavel}~\snm{Sountsov}},
\author[C]{\fnms{Lionel}~\snm{Riou-Durand}},
\author[D]{\fnms{Aki}~\snm{Vehtari}},
\author[E]{\fnms{Andrew}~\snm{Gelman}}



\runauthor{Margossian et al}

\address[A]{Center for Computational Mathematics, Flatiron Institute\printead[presep={,\ }]{e1}}
\address[B]{Google Research}
\address[C]{Laboratoire de Math\'ematiques de l'INSA Rouen Normandie}
\address[D]{Department of Computer Science, Aalto University}
\address[E]{Department of Statistics and Political Science, Columbia University}


\end{aug}

\begin{abstract}
  Recent developments in parallel Markov chain Monte Carlo (MCMC) algorithms allow us to run thousands of chains almost as quickly as a single chain, using hardware accelerators such as GPUs.
  While each chain still needs to forget its initial point during a warmup phase, the subsequent sampling phase can be shorter than in classical settings, where we run only a few chains.
  To determine if the resulting short chains are reliable, we need to assess how close the Markov chains are to their stationary distribution after warmup.
  The potential scale reduction factor $\widehat R$ is a popular convergence diagnostic but unfortunately can require a long sampling phase to work well.
  We present a nested design to overcome this challenge and a generalization called \textit{nested} $\widehat R$.
  This new diagnostic works under conditions similar to $\widehat R$ and completes the workflow for GPU-friendly samplers.
  In addition, the proposed nesting provides theoretical insights into the utility of $\widehat R$, in both classical and short-chains regimes.

\end{abstract}

\begin{keyword}[class=MSC]
\kwd[Primary ]{62-08}
\kwd{Computational methods for problems pertaining to statistics}
\end{keyword}

\begin{keyword}
\kwd{Markov chain Monte Carlo}
\kwd{parallel computation}
\kwd{convergence diagnostics}
\kwd{Bayesian inference}
\kwd{$\widehat R$ statistic}
\end{keyword}

\end{frontmatter}

\section{Introduction}

Over the past decade, much progress in computational power has come from special-purpose single-instruction multiple-data processors such as GPUs.
This has motivated the development of GPU-friendly Markov chain Monte Carlo (MCMC) algorithms designed to efficiently run many chains in parallel \citep[e.g.,][]{Lao:2020, Hoffman:2021, Sountsov:2021, Hoffman:2022, Riou-Durand:2022b}.  
These methods often address shortcomings in pre-existing samplers designed with CPUs in mind:
for example, ChEES-HMC \citep{Hoffman:2021} is a GPU-friendly alternative to the popular but control-flow-heavy no-U-turn sampler (NUTS) \citep{Hoffman:2014}.
With these novel samplers, we can sometimes run thousands of chains almost as quickly as a single chain on modern hardware \citep{Lao:2020}.

In practice, MCMC operates in two phases: a warmup phase that reduces the bias of the Monte Carlo estimators and a sampling phase during which the variance decreases with the number of samples collected.
There are two ways to increase the number of samples: run a longer sampling phase or run more chains.
Practitioners often prefer running a longer sampling phase because each chain needs to be warmed up and so the total number of warmup operations increases linearly with the number of chains.
However, with GPU-friendly samplers, it is possible to efficiently run many chains in parallel.
As a result, the higher computational cost for warmup only marginally increases the algorithm's runtime \citep{Lao:2020}.
It is then possible to trade the length of the sampling phase for the number of chains \citep{Rosenthal:2000}.
When running hundreds or thousands of chains, we can rely on a much shorter sampling phase than when running only 4 or 8 chains.
This defines the \textit{many-short-chains} regime of MCMC.

The length of the warmup phase is a crucial control parameter of MCMC. 
If the warmup is too short, the chains will not be close enough to their stationary distribution and the first iterations of the sampling phase will have an unacceptable bias.
On the other hand, if the warmup is too long, we waste precious computation time.
Both concerns are exacerbated in the many-short-chains regime.
A large bias at the beginning of the sampling phase implies that the entire (short) chain carries a large bias, but running a longer warmup comes at a relatively high cost, since the warmup phase dominates the computation.

To check if the warmup phase is sufficiently long, practitioners often rely on convergence diagnostics \citep{Cowles:1996, Robert:2004, Gelman:2011, Gelman:2013}.
Here, several notions of convergence for MCMC may be considered.
When studying the warmup length, the emphasis is often on the total variation distance $D_\text{TV}$ between the distribution of the first sample obtained after warmup and the target distribution $p$.
Colloquially, has the Markov chain sufficiently approached its stationary distribution during warmup?
We may also consider the bias of the Monte Carlo estimator, for any quantity of interest, and check that this bias is small, even negligible, before we start sampling.
Convergence in $D_\text{TV}$ relates to convergence in bias, though the two notions are not equivalent; see \citet[Proposition 3]{Roberts:2004}.
Unfortunately, neither $D_\text{TV}$ nor the bias can be measured, and so these quantities must be monitored by indirect means.

In the multiple-chains setting, the most popular convergence diagnostic may well be the potential scale reduction factor $\widehat R$ \citep{Gelman:1992, Brooks:1998, Vehtari:2021}.
%
The driving idea behind $\widehat R$ is to compare multiple independent chains initialized from an overdispersed distribution and check that, despite the different initialization, each Markov chain still produces Monte Carlo estimators in close agreement.
In other words, we check how well the Markov chain ``forgets'' its starting point with the understanding that, once the influence of the initialization vanishes, the chain must have reached its stationary distribution and the bias decayed to 0.

In this paper, we formally describe ``forgetfulness'' as the decay of a \textit{nonstationary variance} (to be defined), show how to monitor it, and demonstrate how the nonstationary variance relates to bias decay.
Furthermore, we show that, in the many-short-chains regime, $\widehat R$ may do a poor job monitoring the nonstationary variance and we propose a generalization of $\widehat R$, called \textit{nested} $\widehat R$, to adress this shortcoming.

\subsection{A motivating problem} \label{sec:motivation}

\begin{figure}
    \centerline{
    \includegraphics[width = 2.5in]{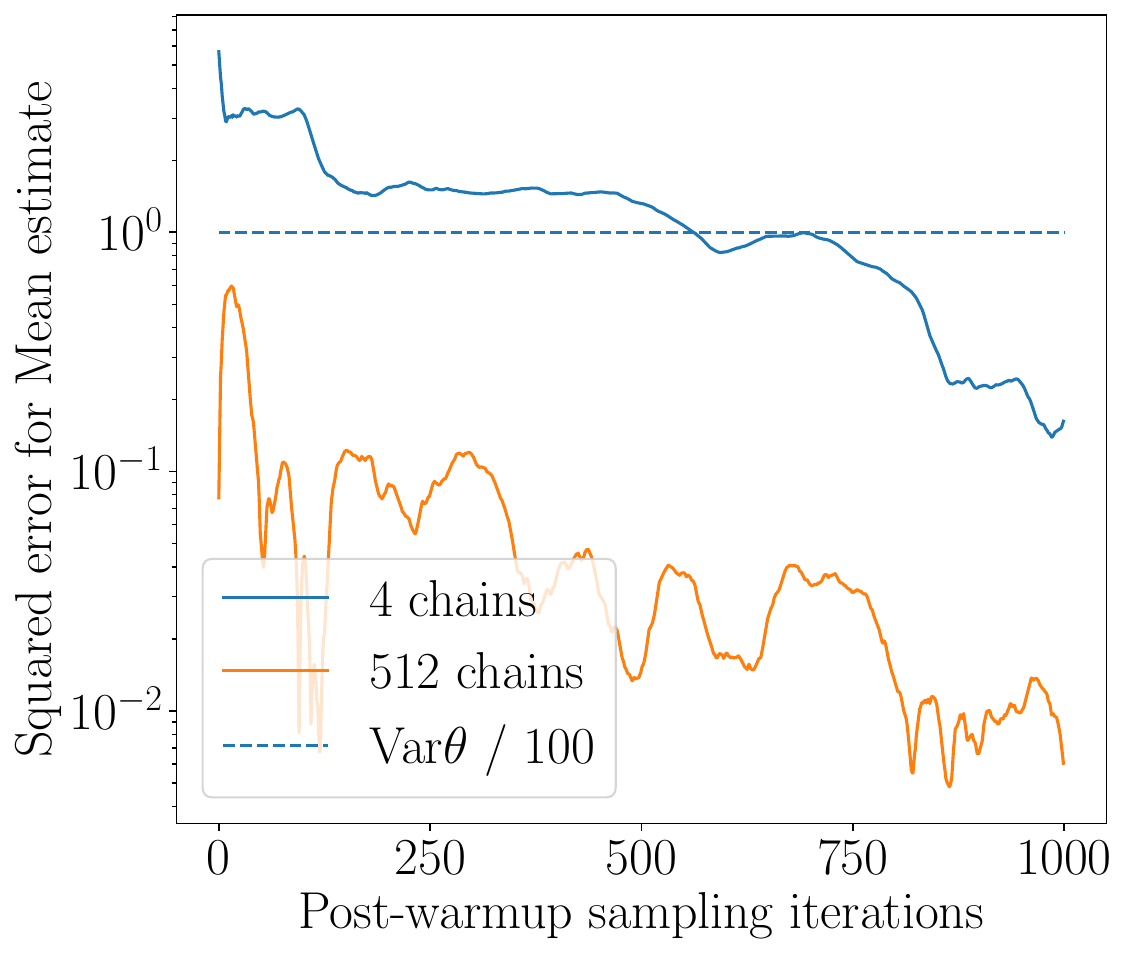}
    \includegraphics[width = 2.6in]{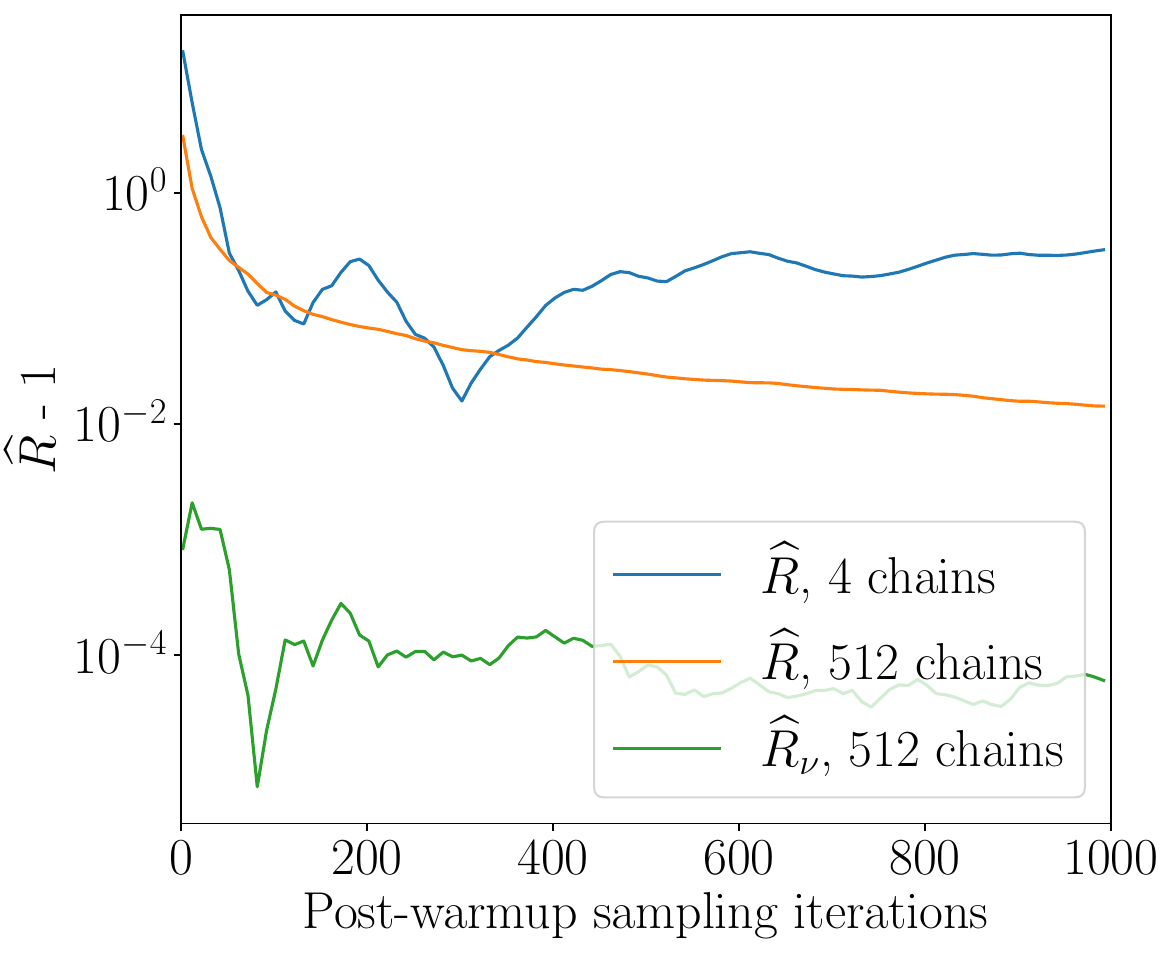}}
    \caption{\em Monte Carlo error when estimating $\mathbb E (\theta_1)$ in the Rosenbrock distribution. (left) Squared error after warmup. When using many chains, a single iteration per chain suffices to achieve a target ESS of 100.
    (right) $\widehat R$ computed with 4 or 512 chains. Many iterations are required to diagnose convergence (i.e. $\widehat R \le 1.01$), even when running 512 chains. In this paper, we introduce nested $\widehat R$, denoted $\nR$---here, evaluated with 4 clusters of 128 subchains---as an alternative to the traditional $\widehat R$ to diagnose convergence in the many-short-chains regime.
    }
    \label{fig:error2_banana}
\end{figure}

We first examine the behavior of $\widehat R$ in the classic regime and the many-short-chains regime of MCMC.
Consider the Rosenbrock distribution $p(\theta_1, \theta_2)$,
\begin{equation} \label{eq:rosenbrock}
    \theta_1 \sim \text{normal}(0, 1) \ ; \ \
    \theta_2 \mid \theta_1 \sim \text{normal}(0.03(\theta_1^2 - 100), 1),
\end{equation}
and suppose we wish to estimate $\mathbb E_p (\theta_1)$, and achieve a squared error below $\text{Var}_p (\theta_1) / 100$.
This corresponds to the expected squared error attained with 100 independent samples from $p(\theta_1, \theta_2)$.
We run ChEES-HMC \citep{Hoffman:2021} on a T4 GPU (available for free on Google Colaboratory\footnote{\url{https://colab.google/}}), first using 4 chains and then using 512 chains.
We discard the first 100 iterations as part of the warmup phase and then run 1000 sampling iterations.
On a GPU, running ChEES-HMC with 512 chains takes $\sim$20\% longer than running 4 chains.\footnote{Run time is evaluated using the Python command \texttt{\%\%timeit}, which reports an average $\pm$ standard deviation run time of: $2.42\pm0.03s$ for 4 chains and $3.07\pm0.04s$ for 512 chains, evaluated by running the code 7 times.} 
Using 4 chains, we require a sampling phase of $\sim$600-700 iterations to achieve our target precision (Figure~\ref{fig:error2_banana}, left).
With 512 warmed-up chains, the target error is attained after 1 sampling iteration.

Next we compute the convergence diagnostic $\widehat R$, and we check whether or not $\widehat R$ is close to 1, for example below the threshold 1.01 proposed by \citet{Vehtari:2021}.
But even after running 1000 sampling iterations, we find that $\widehat R > 1.01$ (Figure~\ref{fig:error2_banana}, right), whether we run 4 chains or 512 chains.
Clearly, the computation cost to reduce $\widehat R$ to 1.01 far exceeds the cost to achieve our target precision, especially when running many chains.
Figure~\ref{fig:error2_banana} further suggests that, whether we run 4 chains or 512 chains, $\widehat R$ decays to 1 at the same rate, even if $\widehat R$ is noisier when computed with fewer chains.

In this paper, we show that for $\widehat R$ to go to 1, the variance of the Monte Carlo estimator generated by a \textit{single} chain must decrease to 0.
This criterion cannot be met if each individual chain is short.
Crucially, $\widehat R$ is a measure of mixing of the chains, which is a separate question than whether the final Monte Carlo estimator, obtained by averaging all the chains, has an acceptable error.  For a simple example, consider a large number of chains started at random positions from the target distribution (or, equivalently, warmed up long enough to approximate independent draws from the target). Inference can be fine after a single post-warmup iteration, even though it could take a long time for the chains to mix.

\subsection{Main Ideas: nonstationary variance and nested $\widehat R$} \label{sec:ideas}

To address this issue, we introduce \textit{nested} $\widehat R$, denoted $\nR$.
The key idea is to compare clusters of chains or \textit{superchains} rather than individual chains.
In order to still track the influence of the initialization, we require all the chains within a superchain to start at the same location.
After a sufficiently long warmup, $\nR$ decays to 1 with the number of subchains, even when each chain remains short.
We show this on our motivating example, where we split the 512 chains into 4 superchains of 128 subchains; see the right panel of Figure~\ref{fig:error2_banana}.

Our approach can be motivated by an analysis of the Monte Carlo estimator's variance.
Consider a state space $\Theta$ over which the target distribution $p$ is defined, and suppose we want to estimate $\mathbb E (f(\theta))$, where $\theta \in \Theta$ and $f$ maps $\theta$ to a univariate variable.
In practice were are interested in multiple such functions $f$.
Let $\bar f^{(1)}$ be the Monte Carlo estimator generated by a single Markov chain and let $\Gamma$ be the distribution of $\bar f^{(1)}$; that is,
\begin{equation}
  \bar f^{(1)} \sim \Gamma.
\end{equation}
$\Gamma$ is characterized by an initial draw from a starting distribution, $\theta_0 \sim p_0$,
and then $\gamma$, the construction of the Markov chain starting at $\theta_0$.
The process $\gamma$ includes the warmup phase (which is discarded) and the sampling phase.
Then by the law of total variance
\begin{equation} \label{eq:variance-decomposition}
    \text{Var}_\Gamma \ \bar f^{(1)} = \underbrace{\text{Var}_{p_0} \left [ \mathbb E_\gamma (\bar f^{(1)} \mid \theta_0) \right]}_{\text{nonstationary variance}}
    + \underbrace{\mathbb E_{p_0} \left [ \text{Var}_\gamma (\bar f^{(1)} \mid \theta_0) \right]}_{\text{persistent variance}}.
\end{equation}
We call the first term on the right side of \cref{eq:variance-decomposition} the \textit{nonstationary variance} and propose to use it as a formal measure of how well the Markov chain forgets its starting point.
Indeed, if the warmup phase is sufficiently long, the initialization should only have a negligible influence on $\mathbb E (\bar f^{(1)})$ and so the nonstationary variance should be close to 0 (even when the sampling phase is short).
Furthermore, the nonstationary variance acts as a ``proxy clock'' for the bias, meaning the time at which the nonstationary variance decays can approximate the time at which the squared bias decays.
We will illustrate, in an example, that both the squared bias and the nonstationary variance can decay at the same rate, and so that the nonstationary variance can be monitored in order to establish whether the warmup phase is sufficiently long for the bias to decay.

The second term on the right side of \cref{eq:variance-decomposition} is the \textit{persistent variance}:
even if the chain reaches stationarity during warmup, the persistent variance remains large and can only decay after a long sampling phase.
Fortunately, the persistent variance can be averaged out and its contribution to the final Monte Carlo estimator decreases to 0 as we increase the number of chains.
Therefore, in the many-short-chains regime, the quantity of primary interest is the nonstationary variance due to its relationship to the squared bias, which cannot be reduced by increasing the number of chains.



\begin{figure}
\begin{center}
\begin{tikzpicture}
 [
    Box/.style={rounded corners, draw=black!, fill=green!0, thick, minimum size=10mm,
                      text width=2cm, align=center},
    Cyan/.style={rounded corners, draw=cyan!, fill=green!0, thick, minimum size=10mm,
                      text width=2cm, align=center},
    Orange/.style={rounded corners, draw=orange!, fill=green!0, thick, minimum size=10mm,
                            text width=2.35cm, align=center},
    Gray/.style={rectangle, draw=black!, fill=gray!35, thick, minimum size=1mm},
    Round/.style={circle, draw=black!, fill=green!0, thick, minimum size=1mm},
    Empty/.style={, draw=white!, fill=green!0, minimum size=0mm}
 ]

  \node[Box] (error2) at(0,0) {\small Squared Error};
  \node[Empty] (=) at (1.5, 0) {=};
  \node[Orange] (variance) at(3.25, 0) {\small Squared Bias};
  \node[Empty] (+) at (4.85, 0) {+};
  \node[Orange] (variance) at(6.5, 0) {{\small Nonstationary Variance}};
  \node[Empty] (+) at (8.15, 0) {+};
  \node[Cyan] (variance) at(9.65, 0) {\small Persistent Variance};
  \node[Empty] (nRhat) at(8., -0.85) {$\underbrace{\hspace{5.5cm}}$};
  \node[Empty] (nRhatMeasure) at (8, -1.5) {\small Monitored by $\nR$};
  \node[Empty] (time) at(4.85, 0.85) {$\overbrace{\hspace{6cm}}$};
  \node[Empty] (Decay) at (4.85, 1.5) {\small Decays as the chains converge};
  \node[Empty, text width=4cm] (DecayM) at (10, 1.75) {\small Decays as the number of subchains increases};
  \node[Empty] (mchains) at (9.65, 0.85) {$\overbrace{\hspace{2.25cm}}$};
  \end{tikzpicture}
  \caption{\em Expected squared error for the Monte Carlo estimator generated by a superchain.
  The estimator is obtained by taking the sample mean of a superchain, after discarding the warmup phase.}
  \label{fig:summary}
\end{center}
\end{figure}

To understand the behavior of $\nR$, we must now analyze the variance of the Monte Carlo estimator returned by a superchain, i.e. a cluster of chains initialized at the same point and then run independently.
We show that in this case, the persistent variance decreases linearly with the number of subchains and so becomes small if we run a large number of chains.
On the other hand, the nonstationary variance remains unchanged, reflecting the unchanged bias.
$\nR$ can then effectively monitor the nonstationary variance of the Markov chain, even when each individual chain is short, provided we run many chains.
Figure~\ref{fig:summary} summarizes the driving idea behind $\nR$.

\subsection{Outline and contributions}

The paper is organized as follows:
\begin{itemize}

\item We define $\nR$ as a generalization of $\widehat R$ and analyze the variance of the Monte Carlo estimator produced by a superchain.
We showcase $\nR$'s ability to effectively monitor the nonstationary variance.
In the edge case where each Markov chain contains one sample, we propose an exact correction to remove the influence of the persistent variance.

\item We then consider a setting where the behavior of $\nR$ can be analyzed analytically: the continuous time limit of MCMC when targeting a Gaussian distribution, attained by a Langevin diffusion.
In this context, we show that the nonstationary variance decays at the same exponential rate as the squared bias.

\item We empirically study the variance of $\nR$ at various phases of the MCMC warmup.
Given a fixed total number of chains, we find that no choice of superchain size uniformly minimizes the variance of $\nR$ and we provide some recommendations.
Finally, we showcase the use of $\nR$ across a range of Bayesian modeling problems, including ones where mixing is slow, the target is multimodal, or the parameter space is moderately high-dimensional ($d=501$). 

\end{itemize}

Throughout the paper, details of proofs are relegated to Appendix~\ref{app:proofs}.
Code in \texttt{Python} to reproduce the figures and experiments in this paper can be found at \url{https://github.com/charlesm93/nested-rhat}.
In addition, an implementation of $\nR$ is available in the \texttt{R} package \texttt{posterior} \citep{posterior}. 

\subsection{Related work}

There has recently been a renewed interest in $\widehat R$ and its practical implementation \citep{Vehtari:2021, Vats:2021, Moins:2022}, with an emphasis on the classical regime of MCMC with 8 chains or fewer.
$\widehat R$ computes a ratio of two standard deviations and is straightforward to evaluate.
However, when applied to samples which are neither independent nor identically distributed---as is the case for MCMC---it becomes difficult to understand which quantity $\widehat R$ measures.
It is moreover unclear how to choose a threshold for $\widehat R$ to decide if the Markov chains have converged.
These issues were recently raised by \citet{Vats:2021} and \citet{Moins:2022}, who studied the convergence of $\widehat R$ in the limit of infinitely long chains.
But such an asymptotic analysis tacitly applies to stationary Markov chains and convergence of the Monte Carlo estimator itself during the sampling phase, rather than convergence in $D_\text{TV}$ or in bias during the warmup.
As prescribed by the many-short-chains regime, we took limits in another direction: an infinite number of finite nonstationary Markov chains.
This perspective elicits the nonstationary variance, sheds light on the advantages and limitations of $\widehat R$, and motivates $\nR$.

Beyond running a long warmup phase, many strategies have been proposed to control the bias of Monte Carlo estimators.
Examples include annealed importance sampling \citep{Neal:2001}, sequential Monte Carlo \citep{DelMoral:2006}, and
unbiased MCMC \citep{Glynn:2014, Jacob:2020}.
The latter relies on transition kernels which allow pairs of Markov chains to couple after a finite time.
Once a coupling occurs, we can construct unbiased estimators of expectation values.
In this sense, the \textit{coupling time} replaces the traditional warmup phase.
Designing such transition kernels with short coupling times is an active area of research \citep{Heng:2019, Jacob:2020, Nguyen:2022}, but at present it is not always possible to find a kernel that couples quickly. 
For example, Hamiltonian Monte Carlo \citep[HMC;][]{Neal:2012, Betancourt:2018} methods using many integration steps per proposal are often the only viable option for sampling from poorly conditioned high-dimensional posteriors over a continuous space.
Unfortunately the coupling HMC kernel of \citet{Heng:2019} only couples quickly if the problem is sufficiently well conditioned that HMC converges rapidly with a relatively small number of integration steps per proposal.

Methods such as Stein thinning \citep{Riabiz:2022} can remove strongly biased samples produced during the early stages of MCMC, and in some sense automatically discard a warmup phase \citep{South:2021}.
However, Stein methods can be computationally expensive and may not scale well in high dimensions.
Furthermore, if the total length of the Markov chain is too short, no thinning can remove the bias incurred by the initialization---a problem which, arguably, $\nR$ can detect and so the proposed diagnostic may be used conjointly with Stein methods.

\section{Nested $\widehat R$} \label{sec:nRhat}

The motivation for $\nR$ is to construct a diagnostic which does not conflate poor convergence and short sampling phase.
$\nR$ works by comparing Monte Carlo estimators whose persistent variance decreases with the number of chains.
However, the bias does not decrease with the number of chains and this should be reflected in the nonstationary variance.
Our solution is introduce superchains, which are groups of chains initialized at the same point and then run independently.

Consider an MCMC process over a state space $\Theta$, which converges in $D_\text{TV}$ to $p$ for any choice of initialization; 
\citet{Roberts:2004} review conditions under which a Markov chain converges in a general state space.
We denote the starting distribution $p_0$.
Each Markov chain is initialized at a point drawn from $p_0$, and comprises $\mathcal W$ warmup iterations, which are discarded, and $N$ sampling iterations.
The warmup phase can be successive applications of a fixed transition kernel or it can involve an adaptation of the transition kernel \citep{Andrieu:2008}.
In the latter case, we require that the Markov chain still converges in $\text{D}_\text{TV}$ to $p$; this can be achieved by choosing an adaptation scheme that preserves the stationary distribution \citep[e.g.,][]{Gilks:1994, Hoffman:2022} or by freezing the transition kernel after a finite number of warmup iterations.
The transition kernel stays fixed during the sampling phase.

So far, we have considered a standard setup for MCMC.
We now introduce the concept of superchain.
\begin{definition}
    (Superchain) We call superchain a collection of $M$ Markov chains. Furthermore,
    \begin{enumerate}
        \item We say the superchain is {\em constrained} if all its subchains are initialized at the same point,
        \item We say the superchain is {\em naive} if all its subchains are initialized independently.
    \end{enumerate}
    In either case, conditional on the initial point, the Markov chains are independent.
\end{definition}
Throughout the paper, we primarily focus on constrained superchains and drop the ``constrained'' prefix for brevity.
Naive superchains are more straightforward to construct and serve as a benchmark.

We generate $K$ superchains, each compromised of $M$ subchains.
In the case $M=1$ where each superchain only contains a single subchain, we recover the classic regime of MCMC.
%
We denote as $\theta^{(nmk)}$ the $n$th draw from chain $m$ of superchain $k$.
Suppose our goal is to estimate $\mathbb E (f(\theta))$ for some function $f$ that maps $\theta$ to a univariate variable and assume furthermore that $\text{Var}_p f(\theta) < \infty$.
The sample mean of each superchain is
\begin{equation}
    \label{eq:nrhat}
    \bar f^{(..k)} = \frac{1}{MN} \sum_{m = 1}^M \sum_{n = 1}^N f \left (\theta^{(nmk)} \right).
\end{equation}
Moving forward, we write $f^{(nmk)} = f (\theta^{(nmk)})$ to alleviate the notation.
The final Monte Carlo estimator is obtained by averaging all the superchains,
\begin{equation}
  \bar f^{(\cdot \cdot \cdot)} = \frac{1}{K} \sum_{k = 1}^K \bar f^{(\cdot \cdot k)}.
\end{equation}
We now define $\nR$.
\begin{definition} \label{def:nRhat}
Consider the estimator $\nB$ of the between-superchain variance,
\begin{equation} \label{eq:Bhat}
  \nB = \frac{1}{K - 1} \sum_{k = 1}^K \left (\bar f^{(\cdot \cdot k)} - \bar f^{(\cdot \cdot \cdot)} \right)^2,
\end{equation}
and the estimator $\nW$ of the within-superchain variance,
\begin{eqnarray}  \label{eq:What}
    \nW & = & \frac{1}{K} \sum_{k = 1}^K \left (\tilde B_k  + \tilde W_k \right), 
\end{eqnarray}
where we require either $N > 1$ or $M > 1$, and we introduce the estimators $\tilde B_k$ of the between-chain variance and $\tilde W_k$ of the within-chain variance; that is
\begin{eqnarray*}
  \tilde B_k & \triangleq & \begin{cases}
    \frac{1}{M - 1} \sum_{m = 1}^M \left(\bar f^{(\cdot mk)} - \bar f^{(\cdot \cdot k)} \right)^2 & \ \text{if \ } M > 1, \\
    0 & \ \text{if \ } M = 1,
  \end{cases} \\
  \tilde W_k & \triangleq & \begin{cases}
    \frac{1}{M} \sum_{m = 1}^M \frac{1}{N - 1} \sum_{n = 1}^N \left( f^{(nmk)} - \bar f^{(\cdot mk)} \right)^2 & \ \text{if \ } N > 1, \\
    0 & \ \text{if \ } N = 1.
  \end{cases}
\end{eqnarray*}
Then $\nR$ is the ratio between the total sample standard deviation across all super chains and the average within-superchain sample standard deviation, that is
 \begin{equation}
   \nR \triangleq \sqrt{\frac{\nW + \nB}{\nW}} = \sqrt{1 + \frac{\nB}{\nW}}.
 \end{equation}
\end{definition}

\ \\[-0.25in]
\begin{remark}
%
In the edge case where $M = 1$, $\nR$ reduces to $\widehat R$.\footnote{The original $\widehat R$ uses a slightly different estimator for the within-chain variance $\tilde W_k$ when computing the numerator in $\widehat R$.
There the sample variance $\tilde W_k$ is scaled by $1 / N$, rather than $1 / (N - 1)$.
This explains why occasionally $\widehat R < 1$.
This is of little concern when $N$ is large, but we care about the case where $N$ is small, and we therefore adjust the $\widehat R$ statistic slightly.}
\end{remark}

A common practice to assess convergence of the chains is to check that $\widehat R \le 1 + \epsilon$ for some $\epsilon > 0$.
Recommended choices of $\epsilon$ have evolved over time, starting at $\epsilon = 0.1$ \citep{Gelman:1992} and recently using the more conservative value $\epsilon = 0.01$ \citep{Vehtari:2021}.
In the proposed nested design, we now check whether $\nR \le 1 + \epsilon$ and therefore,
\begin{equation} \label{eq:tolerance}
    \nB \le 2 \epsilon \ \nW + \mathcal O(\epsilon^2).
\end{equation}
This inequality establishes a tolerance value for $\nB$, scaled by the within-superchain-variance $\nW$.
Moreover, we want to check that despite the distinct initialization and seed, the sample means of each superchain are in good agreement, as measured by $\nB$.

\section{Properties of nested $\widehat R$}

We now analyze the properties of $\nR$ theoretically and illustrate its use on several examples.
For all formal statements, we assume the superchains are constrained, meaning all their subchains are initialized at the same point.

\subsection{Which quantity does $\nR$ measure?}

To answer this question, we consider the asymptotics of $\nR$ along $K$, the number of superchains.
Applying the law of large numbers,
\begin{equation}
  \nB \overset{\text{a.s}}{\underset{K \to \infty}{\longrightarrow}} \B \triangleq \text{Var}_\Gamma (\bar f^{(\cdot \cdot k)}).
\end{equation}
Then by the law of total variance,
\begin{equation} \label{eq:nB}
  \B = \text{Var}_{p_0} \left [ \mathbb E_\gamma (\bar f^{(\cdot \cdot k)} \mid \theta^k_0)  \right]
  + \mathbb E_{p_0} \left [ \text{Var}_\gamma (\bar f^{(\cdot \cdot k)} \mid \theta^k_0) \right].
\end{equation}
Because the chains are identically distributed and conditionally independent, we have
\begin{eqnarray}
  \mathbb E_\gamma (\bar f^{(\cdot \cdot k)} \mid \theta^k_0) & = & \mathbb E_\gamma (\bar f^{(\cdot m k)} \mid \theta^k_0) \\
  \text{Var}_\gamma (\bar f^{(\cdot \cdot k)} \mid \theta^k_0)
  & = & \frac{1}{M} \text{Var}_\gamma (\bar f^{(\cdot m k)} \mid \theta^k_0).
\end{eqnarray}
Plugging these results back into \eqref{eq:nB} yields the following result.
\begin{theorem}  \label{thm:nB}
Consider a constrained superchain initialized at $\theta^k_0 \sim p_0$ and made of $M$ chains. Then the variance of the superchain's sample mean is
\begin{equation} \label{eq:nB}
  \B = \underbrace{\mathrm{Var}_{p_0} \left [ \mathbb E_\gamma (\bar f^{(\cdot m k)} \mid \theta^k_0) \right]}_\text{nonstationary variance}
  + \underbrace{\frac{1}{M} \mathbb E_{p_0} \left [ \mathrm{Var}_\gamma (\bar f^{(\cdot m k)} \mid \theta^k_0) \right]}_\text{persistent variance}.
\end{equation}
\end{theorem}

\begin{remark}
From Theorem~\ref{thm:nB}, we immediately see that the constrained superchain has the same nonstationary variance as a single chain, but the persistent variance decreases linearly with $M$.
\end{remark}

\begin{figure}[!]
    \centering
    \includegraphics[width=4in]{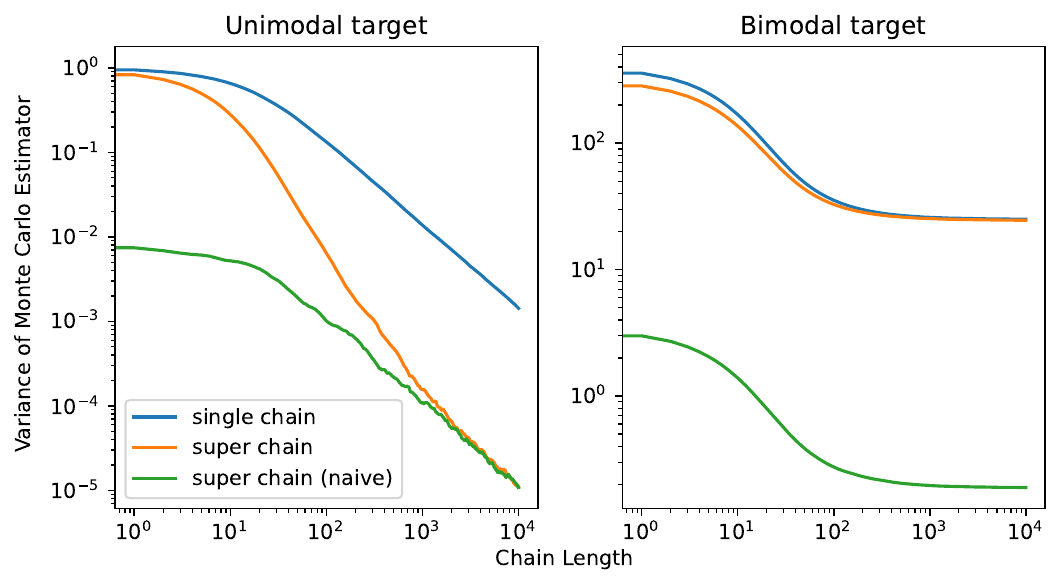}
    \caption{\em Variance of Monte Carlo estimators, $\bar f^{(\cdot \cdot k)}$, constructed using a single chain, a constrained superchain of 1024 subchains initialized at the same point, or a naive superchain of 1024 independent subchains.
    When the Markov chains converge, the variance of a superchain decreases to the variance of a naive superchain.
    This transition occurs because the nonstationary variance decays to 0.
    If the Markov chains do not converge, the variance of a superchain stays large.
    }
    \label{fig:transition}
\end{figure}

We demonstrate the behavior of $\B$ on two target distributions:
\begin{itemize}
    \item A standard Gaussian, $p = \text{normal}(0, 1)$, with initial distribution, $p_0 = \text{normal}(2, 1)$.
    \item An unbalanced mixture of two Gaussians, $p = 0.3 \,\text{normal}(-10, 1) + 0.7 \,\text{normal}(10, 1)$, with initial distribution, $p_0 = \text{normal}(0, 20)$.
\end{itemize}
As benchmarks, we use $B$, the variance of the Monte Carlo estimator generated by a single chain, and $\tilde \B$, the variance of the Monte Carlo estimator generated by a naive superchain.
Because there are no initialization constraints, the nonstationary variance of the naive superchain also scales as $1 / M$ (but the bias stays the same!).
Indeed,
\begin{eqnarray}
    \tilde \B & = & \text{Var}_\Gamma(\bar f^{(\cdot \cdot k)}) \nonumber \\
    & = & \frac{1}{M} \text{Var}_\Gamma(\bar f^{(\cdot m k)}) \nonumber \\
    & = & \frac{1}{M} \text{Var}_{p_0} \left [ \mathbb E_\gamma (f^{(\cdot m k)} \mid \theta^{mk}_0) \right ] + \frac{1}{M} \mathbb E_{p_0} \left [ \text{Var}_\gamma (f^{(\cdot m k)} \mid \theta^{mk}_0) \right ], \\
    & = & \frac{1}{M} \text{Var}_{p_0} \left [ \mathbb E_\gamma (f^{(\cdot m k)} \mid \theta^{k}_0) \right ] + \frac{1}{M} \mathbb E_{p_0} \left [ \text{Var}_\gamma (f^{(\cdot m k)} \mid \theta^{k}_0) \right ]
\end{eqnarray}
where the second line follows from the independence of the Markov chains (something that does not hold for constrained superchains) and the third line from the law of total variance.
In the final line, we replaced $\theta^{mk}_0$ with $\theta^k_0$---seing both variables are drawn from $p_0$---to recover the terms in \eqref{eq:nB}.

We run Hamiltonian Monte Carlo on a GPU and use $M = 1028$ chains for each superchain.
Figure~\ref{fig:transition} shows the results.
When the chains converge, $\B$ first behaves like $B$ and then transitions to behaving like $\tilde \B$, as the nonstationary variance decays.
When the chains do not converge, the transition does not occur and $\B$ stays large, giving a clear indication that the chains do not converge.
One drawback is that the constrained superchain has a larger variance than the naive superchain, especially during the early stages of MCMC; however, if the Markov chains converge, the nonstationary variance vanishes, and the expected squared error is dominated by the persistent variance, which is the same for both types of superchains.




We end this section with two corollaries of Theorem~\ref{thm:nB}.
To show these corollaries, we first require the following lemma on the asymptotic limit of $\nW$.
\begin{lemma} \label{lemma:nW}
  In the limit of an infinite number of superchains,
    \begin{equation}
    \lim_{K \to \infty} \nW = \W \triangleq \mathbb E_{p_0} \text{Var}_\gamma (\bar f^{(\cdot mk)} \mid \theta^k_0)+ W',
  \end{equation}
  where
  \begin{equation*}
  W' \triangleq \begin{cases}
    \frac{1}{N - 1} \sum_{n = 1}^N \text{Var}_\Gamma f^{(nmk)} - \text{Var}_\Gamma \bar f^{(\cdot mk)}
       + (\mathbb E_\Gamma f^{(nmk)})^2 - (\mathbb E_\Gamma f^{(\cdot mk)})^2 & \ \text{if} \ N > 1 \\
       0 & \ \text{if} \ N = 1,
  \end{cases}
\end{equation*}
  
\end{lemma}

The first corollary of Theorem~\ref{thm:nB} examines the behavior of $\nR$ when the samples are in stationarity and all chains, including ones within a superchain, are independent.
This setting approximates the behavior of Markov chains after a long warmup phase, specifically the limit where the Markov chains have converged to their stationary distribution and forgortten their starting point before sampling.

\begin{corollary}  \label{thm:lower-nR}
  Suppose that all chains within the same superchain are independent and initialized from the stationary distribution.
  Assume further that the length of the sampling phase $N > 1$ and the chains have positive autocorrelation.
  Then, 
  \begin{equation}
    \sqrt{1 + \frac{\B}{\W}} 
    \ge \sqrt{1 + \frac{1}{M} \frac{1 - 1 / N}{\mathrm{ESS}_{(1)} + 1 / N}},
  \end{equation}
  where  $\mathrm{ESS}_{(1)}$ is the effective sample size for a single chain of length $N$.
\end{corollary}

From the above, we see that without the nesting design (i.e. $M = 1$ case), $\widehat R$ can only decay to 1 if $\text{ESS}_{(1)}$ is large, because we need to kill off the persistent variance.
This explains the results seen in Figure~\ref{fig:error2_banana} and
echoes the observation by \citet{Vats:2021} that for stationary Markov chains, $\widehat R$ (or rather the quantity measured by $\widehat R$) is a one-to-one map with the ESS.
We emphasize however that this equivalence does not hold when the chains are not stationary.

One persistent problem is that, even though we want to monitor the nonstationary variance to assess convergence, we instead measure the total variance.
Increasing $M$ and $N$ reduces the persistent variance, and makes $\B$ a sharper bound on the nonstationary variance.
The second corollary of Theorem~\ref{thm:nB} tells us that when $N=1$, that is each chain contains a single sample, we can exactly correct for the contribution of the persistent variance. 
The $N=1$ case is is mathematically convenient and takes the logic of the many-short-chains regime to its extreme, with the variance reduction entirely handled by the number of chains.
%
%
\begin{corollary} \label{corr:N=1}
  Suppose that all chains within a superchain start at the same point $\theta_0^k \sim p_0$ and that each chain is made of $M$ subchains.
  Assume further that $N = 1$.
  Then the persistent variance, scaled by $\W$, is
  \begin{equation}
    \frac{\mathbb E_{p_0} \text{Var}_\gamma (\bar f^{(\cdot \cdot k)} \mid \theta^k_0)}{\W} = \frac{1}{M},
  \end{equation}
  and furthermore
  \begin{equation}
    \sqrt{1 + \frac{\B}{\W}} = \sqrt{1 + \frac{1}{M} + 
      \frac{\mathrm{Var}_{p_0} \mathbb E_\gamma (f^{(1 m k)} \mid \theta^k_0)}{
        \mathbb E_{p_0} \mathrm{Var}_\gamma (f^{(1 mk)} \mid \theta^k_0)}}.
  \end{equation}
\end{corollary}
%
%
We can now express a threshold on $\nR$ as a threshold on the nonstationary variance, scaled by the within-superchain  variance.
Conversely, we can express a bound on the standardized nonstationary variance as an equivalent bound on $\nR$.
We demonstrate this approach in Section~\ref{sec:tolerance}.
\begin{remark}
    Such a result is not available for $\widehat R$, which by definition must be computed for $N > 1$.
\end{remark}

One potential limitation of Corollary~3.5 is that it still does not give us a consistent estimator of the nonstationary variance, because of the scaling by the within-superchain variance $\W = \mathbb E_{p_0} \mathrm{Var}_\gamma (f^{(1 mk)} \mid \theta^k_0)$.
However, when assessing the quality of MCMC, it is common to standardize quantities of interest, usually by the variance $\text{Var}_p(f)$, in order to obtain a scale-free measure of information.
This is especially helpful when running diagnostics for a large number of variables, since then the threshold need not (in general) be adjusted for each variable.
Here, $\nR$ is scale-free, in the same manner as $\widehat R$ and the ESS.

Now, the nonstationary variance is not scaled by $\text{Var}_p(f)$, which is unknown, but by the (estimated) within-superchain variance $\W$.
At stationarity, $\W = \text{Var}_p(f)$.
For nonstationary Markov chains, $\nR$ may erroneously report convergence if $\W \gg \text{Var}_p(f)$, as this would shrink $\nR$ to 1, even if the nonstationary variance is large relative to $\text{Var}_p(f)$.
In general however, we expect $\W \le \text{Var}_p(f)$ for nonstationary Markov chains, based on empirical observations on the within-chain variance used for $\widehat R$ \citep{Vehtari:2021}.
Two illustrations of this phenomenon are (i) very early in the warmup phase, when all chains within a superchain concentrate around their common initialization and (ii) for $p$ a multimodal distribution, when all subchains explore the same basin of attraction and so  the within-superchain variance corresponds to the variance around one mode, which is inferior to the total variance of $p$.


\subsection{Bias and nonstationary variance}  \label{sec:reliable-theory}

We have established that $\nR$ monitors the nonstationary variance.
However the primary goal of the warmup phase is to reduce the bias.
In this section, we examine the connection between the nonstationary variance and the bias in an illustrative example.
We first define the idea of reliability for a convergence diagnostic.

\begin{definition}
    For a univariate random variable $f$, we say an MCMC process is $(\delta, \delta')$-reliable for $\nR$ if
    \begin{equation}
        \frac{\B}{\W} \le \delta \implies \frac{(\mathbb E_\Gamma \bar f^{(\cdot m)} -  \mathbb E_p f)^2}{ \mathrm{Var}_p f} \le \delta'.
    \end{equation}
\end{definition}
Since $\widehat R$ is a special case of $\nR$, this also defines reliability for $\widehat R$.
\citet{Gelman:1992} tackled the question of $\widehat R$'s reliability by using an overdispersed initialization.
Here, we provide a formal proof that in the Gaussian case $(\delta, \delta')$-reliability is equivalent to using an initial distribution $p_0$ with a large variance relative to the initial bias.
The Gaussian case provides intuition for unimodal targets and can be a reasonable approximation after rank normalization of the samples \citep{Vehtari:2021}.

Let $p = \text{normal}(\mu, \sigma^2)$.
To approximate a large class of MCMC processes, we consider the solution $(X_t)_{t\ge0}$ of the Langevin diffusion targeting $p$ defined by the stochastic differential equation
\begin{equation} \label{eq:sde}
  d X_t = - \,(X_t - \mu) d t + \sqrt{2 \sigma} \, d W_t,
\end{equation}
where $(W_t)_{t\ge0}$ is a standard Brownian motion. 
The convergence of MCMC toward contin\-u\-ous-time stochastic processes has been widely studied, notably by \citet{Gelman:1997} and \citet{Roberts:1998}, who established scaling limits of random walk Metropolis and the Metropolis adjusted Langevin algorithm toward Langevin diffusions. Similar studies have been conducted for Hamiltonian Monte Carlo and its extensions; see, e.g., \citet{Beskos:2013}, \citet{Riou-Durand:2022}. Typically, the solution of a continuous-time process after a time $T>0$ is approximated by a Markov chain, discretized with a time step $h>0$ and run for $\lfloor T/h \rfloor$ steps. We consider here a Gaussian initial distribution
\begin{equation}
p_0 = \text{normal}(\mu_0, \sigma_0).
\end{equation}
In this setup, the bias and the variance of the Monte Carlo estimator admit an analytical form.
The solution $X_T$ is interpreted as an approximation as $h\rightarrow 0$ of the setting of parallel chains for $\mathcal W=\lfloor T/h \rfloor$ iterations and $N = 1$, i.e., $\theta^{(1mk)}=X_T$.
This scenario is the simplest one to analyze and illustrates an edge case of the many-short-chains regime.

The main result of this section states that the squared bias and scaled nonstationary variance decay at a rate $\sim e^{-2T}$,
providing justification as to why the latter can be used as a proxy clock for the former.
\begin{theorem} \label{lemma:langevin}
Let $(X_t)_{t\ge0}$ be the solution to \eqref{eq:sde}, which describes a Langevin diffusion targeting $p = \text{normal}(\mu, \sigma^2)$, starting from $X_0 \sim p_0 = \text{normal}(\mu_0, \sigma_0^2)$. Then for any warmup time $T>0$, the squared bias is
\begin{equation}
  \left (\mathbb E \bar \theta^{(1 \cdot k)} - \mathbb E_p X \right)^2 = \left (\mathbb E X_T - \mathbb E_p X \right)^2 = (\mu_0 - \mu)^2 e^{-2T}.
\end{equation}
Furthermore,
\begin{equation}  \label{eq:langevin_ratio}
  \frac{\B}{\W} = \frac{\mathrm{Var}_{p_0}\mathbb E (X_T \mid X_0 )}{\mathbb E_{p_0} \mathrm{Var} (X_T \mid X_0 )} = 
  \frac{1}{M} + \frac{\sigma_0^2}{\sigma^2 (e^{2T} - 1)}.
\end{equation}
\end{theorem}
We end this section with a corollary that states that the initialization needs to be sufficiently dispersed for $\nR$ to be reliable.
\begin{corollary} \label{thm:langevin}
  Let $\delta > 0$ and $\delta' > 0$.
  Assume the conditions stated in Theorem~\ref{lemma:langevin}. 
  If $(\mu_0 - \mu)^2 / \sigma^2 \le \delta'$, $\nR$ is trivially $(\delta, \delta')$-reliable.
  If $(\mu_0 - \mu)^2 / \sigma^2 > \delta'$, then
  $\nR$ is $(\delta, \delta')$-reliable if and only if
  \begin{equation}  \label{eq:langevin-lower}
    \sigma_0^2 > \left (\delta - \frac{1}{M} \right) \left (\frac{(\mu - \mu_0)^2}{\delta' \sigma^2} - 1 \right) \sigma^2.
  \end{equation}
\end{corollary}

\begin{remark}
  If $\delta < 1 / M$, then the condition $\B / \W < \delta$ cannot be verified.
  Hence $\nR$ is reliable in a trivial sense: we never erroneously claim convergence because we never claim convergence.
\end{remark}


Appendix~\ref{app:reliability} provides additional results on the reliability of $\nR$: we numerically test the lower bound provided by Corollary~\ref{thm:langevin} on a Gaussian and a mixture of Gaussians, and we then derive a similar bound for $\widehat R$ when $N > 1$.

\subsection{Variance of $\nR$} \label{sec:variance}

In the previous sections, we focused on the quantity measured by $\nR$, attained in the limit of a large number of superchains, that is $K \to \infty$.
We now assume a finite number of superchains.
Even if $\nR$ is $(\delta, \delta')$-reliable, it may be too noisy to be useful.
An example of this arises in the multimodal case, where all $K$ superchains may initialize \textit{by chance} in the attraction basin of the same mode.
In this scenario, $K$ is too small and we drastically underestimate $\B = \text{Var}_\Gamma \bar f^{(\cdot \cdot k)}$.
One would need enough superchains with distinct initializations to find multiple modes and diagnose the chains' poor mixing.

When using $\widehat R$ it seems reasonable to increase the number of chains as much as possible.
The question is more subtle for $\nR$: given a fixed total number of chains $KM$ how many superchains should we run?
In general no choice of $K$, given $KM$, minimizes the variance of $\nB / \nW$ across all stages of MCMC.

\begin{figure}
  \begin{center}
  \includegraphics[width=2.5in]{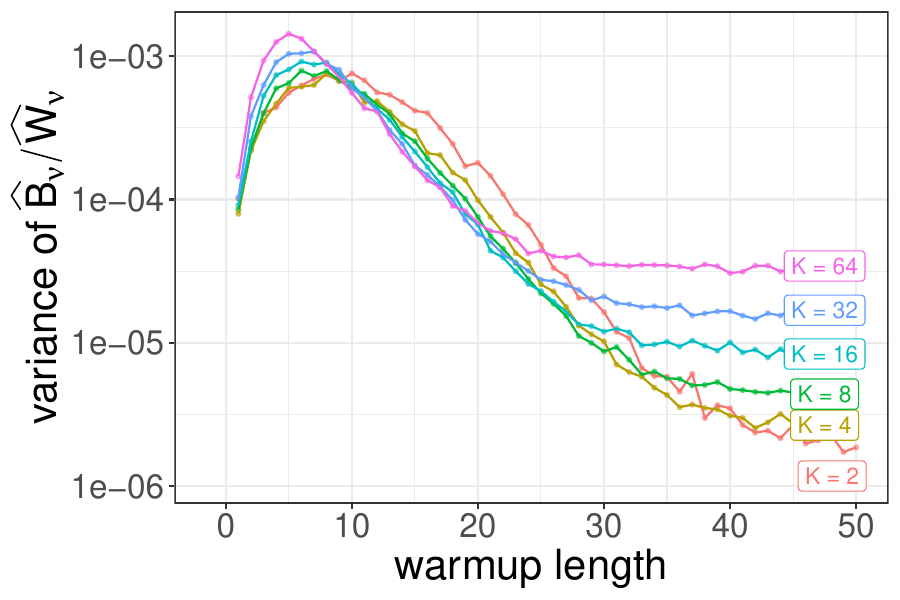}
  \includegraphics[width=2.5in]{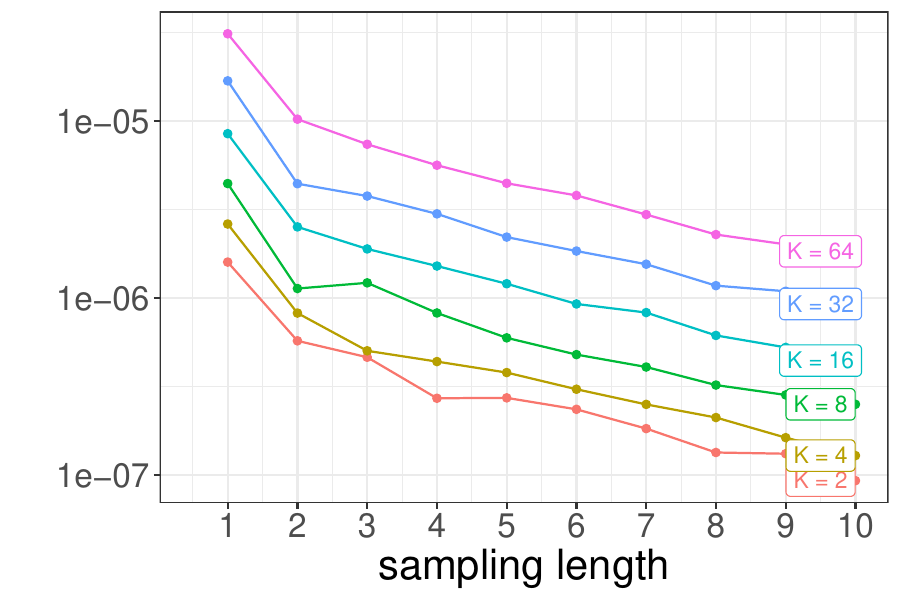}
  \caption{\em Variance of $\nB / \nW$ when running $KM = 2048$ chains split into $K$ superchains.
  (left) $\nB / \nW$ is computed using $N = 1$ samples after a varying warmup length.
  No choice of $K$ uniformly minimizes the variance across all phases of MCMC.
  (right) This time, the chains are stationary.
  Increasing the length of the sampling phase $N$ reduces the variance of $\nB / \nW$.
  }
  \label{fig:variance_K_N}
  \end{center}
\end{figure}

We demonstrate this phenomenon when targeting a Gaussian with $KM = 2048$ chains (left panel of Figure~\ref{fig:variance_K_N}).
Here the choice $K = 2$ minimizes the variance once the chains are stationary but nearly maximizes it during the early stages of the warmup phase.
For $K$ small, all the chains may be in agreement by chance even when they are not close to the stationary distribution.
Increasing $K$ helps avoid this scenario.
On the other hand, a large $K$ results in a large variance once the chains approach their stationary distribution.
This is because the persistent variance remains large if $K$ is large (and therefore $M$ is small), even once the nonstationary variance vanishes.
In other words $\nB$ remains noisy because of a large nuisance term.
We now see the competing forces at play when choosing $K$.
Empirically, we find that several choices of $K$ work well across a collection of problems (Section~\ref{sec:experiments}).

The variance of $\nR$ can be further reduced by increasing the length of the sampling phase $N$. The right panel of 
Figure~\ref{fig:variance_K_N} demonstrates the reduction in variance as $N$ varies from 1 to 10.
For many problems, running 10 more iterations is computationally cheap.
A large $N$ also makes $\B$ a sharper upper bound on the nonstationary variance.
However, if $N > 1$, we cannot use Corollary~\ref{corr:N=1} to exactly characterize and correct for the persistent variance. 

\subsection{Error tolerance and threshold for $\nR$} \label{sec:tolerance}

Ultimately, our goal is to control the error of our Monte Carlo estimators.
Consider the bias-variance decomposition for an estimator returned by a superchain,
\begin{align}
 \mathbb E_\Gamma \left ( (\bar f^{(\cdot \cdot k)} \right. \left. - \mathbb E_p f)^2 \right)
     = \underbrace{(\mathbb E_\Gamma \bar f^{(\cdot \cdot k)} - \mathbb E_p f)^2}_{\text{squared bias}} +  \underbrace{\text{Var}_{p_0} (\mathbb E_\gamma (\bar f^{(\cdot \cdot k)} \mid \theta^k_0))}_{\text{nonstationary variance}} + \underbrace{\mathbb E_{p_0} (\text{Var}_\gamma (\bar f^{(\cdot \cdot k)} \mid \theta^k_0))}_{\text{persistent variance}}.
\end{align}
Much of the MCMC literature focuses on the final term, with the assumption that we have run the Markov chain's warmup ``long enough'' for them to be approximately stationary, at which point the bias (and the nonstationary variance) become negligible.
The error tolerance can then be expressed in terms of the MCSE and the ESS \citep[e.g.][]{Flegal:2008, Gelman:2013, Vats:2019, Vehtari:2022}.
What constitutes an appropriate ESS is problem-dependent and also subject to academic discussion \citep[e.g.,][]{Mackay:2003, Gelman:2011, Vats:2019, Vehtari:2022, Margossian:2023}.

If the Markov chains are not close to stationarity, variance alone cannot characterize the expected squared error, and so we must first check for approximate convergence. 
We may do so by setting a tolerance $\tau$ on the (scaled) nonstationary variance,
\begin{equation}
  \frac{\text{Var}_{p_0} \left ( \mathbb E_\gamma (\bar f^{(\cdot \cdot k)} \mid \theta^k_0) \right)}{\mathbb E_{p_0} \left ( \text{Var}_\gamma (\bar f^{(\cdot \cdot k)} \mid \theta^k_0) \right)} \le \tau.
\end{equation}
Since with $\nR$ we measure the total variance rather than the nonstationary variance, we need to correct for the persistent variance.
When $N = 1$, we can apply Corollary~\ref{corr:N=1}, which tells us that the scaled persistent variance is exactly $1 / M$.
Using the latter, the threshold on $\nR$ is
\begin{equation}
  \nR \le \sqrt{1 + \frac{1}{M} + \tau}.
\end{equation}
If $N > 1$, we may run a large number of subchains for a long enough number of sampling iterations (perhaps more so than needed to achieve our target ESS) to ``kill'' the contribution of the persistent variance.
The threshold on $\nR$ may then be $\sqrt{1 + \tau}$, though this is likely conservative since the persistent variance doesn't completely vanish.
Alternatively, we could attempt to calculate the persistent variance using an estimator of the ESS, for example based on the Markov chain's autocorrelation, if $N$ is sufficiently large.

The choice of $\tau$ itself, just like the tolerable expected squared error, depends on the problem.
We propose to make the nonstationary variance---and so, by proxy, the squared bias---small next to the tolerable squared error.
Then the error is dominated by the persistent variance and can be characterized by estimators of the MCSE.
We will demonstrate such an approach on a collection of examples.

\section{Numerical experiments} \label{sec:experiments}

\begin{table}
  \begin{tabular}{l r l}
  \rowcolor{gray!20} {\bf Target} & $d$ & {\bf Description} \vspace{2mm} \\
  Rosenbrock & 2 & A joint normal distribution nonlinearly transformed to have\\ 
  & & high curvature \citep{Rosenbrock:1960}. See Equation~\ref{eq:rosenbrock}. This 
  \\ & & target produces Markov chains with a large autocorrelation. \vspace{2mm} \\
  \rowcolor{gray!20}
  Eight Schools & 10 & The posterior for a hierarchical model of the effect of a test-\\
  \rowcolor{gray!20} & & preparation program for high school students \citep{Rubin:1981}. \\
  \rowcolor{gray!20} & & Fitting such a model with MCMC requires a careful  \\
  \rowcolor{gray!20} & & reparameterization \citep{Papaspiliopoulos:2007}.
  \vspace{2mm} \\
  German Credit & 25 & The posterior for a logistic regression applied to a numerical \\ & & version of the German credit data set \citep{Dua:2017}. \vspace{2mm} \\
  \rowcolor{gray!20}
  Pharmacokinetics & 45 & The posterior for a hierarchical model describing the absorption\\ 
  \rowcolor{gray!20} & & of a drug compound in patients during a clinical trial \\ 
  \rowcolor{gray!20} & & \citep[e.g.,][]{Wakefield:1996, Margossian:2022-torsten}, using data \\ 
  \rowcolor{gray!20} & & simulated over 20 patients. This model uses a likelihood based \\
  \rowcolor{gray!20} & &  on an ordinary differential equation. \\
  Bimodal & 100 & An unbalanced mixture of two well-separated Gaussians.
  With \\ 
  & & standard MCMC, each Markov chain ``commits'' to a single mode, \\
  & & leading to bias sampling.
  Even after a long compute time, the \\
  & & Markov chains fail to converge.
  \vspace{2mm} \\
  \rowcolor{gray!20}
  Item Response & 501 & The posterior for a model to assess students' abilities based on\\
  \rowcolor{gray!20} & & test scores \citep{Gelman:2007}. The model is fitted to the \\
  \rowcolor{gray!20} & & response of 400 students to 100 questions, and the model \\ 
  \rowcolor{gray!20} & & estimates (i) the difficulty of each
   question and (ii) each student's. \\
  \rowcolor{gray!20} & & aptitude.
  This problem has a relatively high dimension.
  \end{tabular}
  \caption{Target distributions for our numerical experiments.}
  \label{tab:target}
\end{table}

We demonstrate an MCMC workflow using $\nR$ on a diversity of applications mostly drawn from the Bayesian literature.
Our focus is on producing accurate estimates for the first moment of the model parameters.
We consider six targets which represent a diversity of applications.
Table~\ref{tab:target} summarises the target distributions with more details available in Appendix~\ref{app:models}.
The bimodal example provides a case where the chains fail to converge after a reasonable amount of compute time.
As our MCMC algorithm, we run ChEES-HMC \citep{Hoffman:2021}, which is an adaptive HMC sampler, designed to run efficiently on GPUs.
ChEES-HMC pools information across all Markov chains to set the tuning parameters for HMC, specifically (i) the step size of the integrator solving Hamilton's equations of motion and (ii) the number of steps the integrator takes, thereby determining the length of the Hamiltonian trajectory.
Fixing (ii) across all chains ensures ensures that, at a given iteration, each chain requires the same amount of compute and so the operation can be efficiently parallelized on a GPU.
The algorithm is implemented in \texttt{TensorFlow Probability} \citep{tfp:2023}.

\subsection{Monte Carlo squared error after convergence} \label{sec:model-error}

We construct a model of the squared error for stationary Markov chains, and use this as a benchmark for the empirical squared error.
At stationarity, the subchains within a superchain are no longer correlated.
A central limit theorem may then be taken along the total number of chains.
Then for $N = 1$, the scaled squared error approximatively follows a $\chi^2$ distribution,
\begin{equation} \label{eq:chi}
  E^2 \triangleq \frac{KM}{\text{Var}_p f} (\bar f^{(1 \cdot \cdot)} - \mathbb E_p f)^2 \overset{\text{approx.}}{\sim} \chi^2_1.
\end{equation}
High-precision estimates of $\mathbb E_p f$ and $\text{Var}_p f$ are computed using long MCMC runs (Appendix~\ref{app:models}).
We use the above approximation to jointly model the expected squared error at stationarity across all dimensions.
This is somewhat of a simplification, since we do not account for the correlations between dimensions.

After a sufficiently long warmup phase, the Markov chains are nearly stationary and $E^2$ approximately follows a $\chi^2_1$ distribution.
This should ideally be reflected by $\nR \approx 1$.
However if the warmup phase is too short and the squared error large because of the Markov chain's bias, we expect to see $\nR \gg 1$.

\subsection{Results when running 2048 chains with $K = 16$ superchains}

Suppose we target an ESS of $\sim$2000.
In this case, we may run 2048 chains, broken into $K = 16$ superchains and $M = 128$ subchains.
For each target distribution, we compute $\nR$ using $N = 1$ draw after warmups of varying lengths, \begin{equation*}
  {\bf W} = (10, 20, 30, \dots, 100, 200, 300, \dots, 1000).
\end{equation*}
${\bf W}$ contains both warmup lengths that are clearly too short to achieve convergence and lengths after which approximate convergence is expected when running HMC.
That way, the behavior of $\nR$ can be examined on both nonstationary and (nearly) stationary Markov chains.
At the end of each warmup window, we record $\nR$ for each dimension and the corresponding $E^2$ using a single draw per chain.
We repeat our experiment 10 times for each model.

While helpful for diagnosing convergence, constraining all subchains within a superchain to start at the same location increases the total Monte Carlo variance (Theorem~\ref{thm:nB}).
We empirically investigate the repercusions of this initialization scheme on the mean squared error (MSE).
Overall, we do not see a drastic difference between using a constrained or a naive superchain, especially as the warmup length approaches 1000  iterations (Figure~\ref{fig:MSE}).
For the Item Response model, the MSE decays faster with a naive superchain on average.
For the hierarchical models (Eight Schools and Pharmacokinetics) the MSE decays more slowly with a naive superchain.
This phenomenon is consistent with past observations that for models with an intricate posterior geometry, one or more chains can get stuck in difficult regions of the probability space, usually well in the tails of the target distribution \citep{Hoffman:2021, DuChe:2023}.
The adaptation strategy of ChEES-HMC is to reduce the step size of HMC to insure that Markov chains stuck in difficult regions can still move forward (``no chains left behind'' heuristic).
Now, ChEES-HMC uses a common transition kernel for all chains, and a conservative step size, while helpful for stuck chains, can be suboptimal for other chains that may already have reached regions where the posterior probability mass concentrates.
Unfortunately, increasing the number of distinct initializations increases the probability that at least one chain gets stuck in a difficult region, particularly during the early warmup.

\begin{figure}
    \centering
    \includegraphics[width=5in]{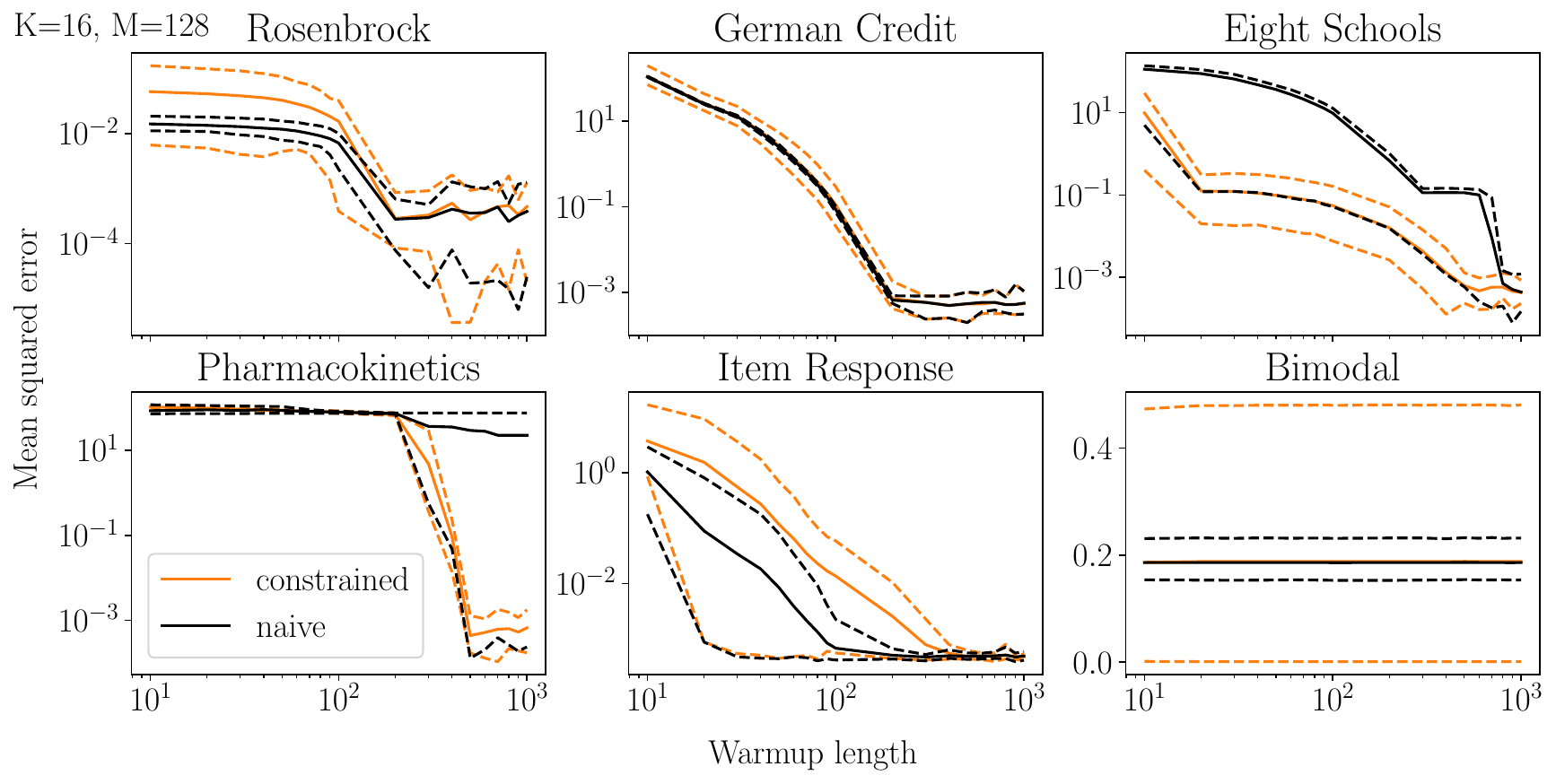}
    \caption{\textit{Mean squared error, scaled by the posterior variance, when using $K=16$ constrained superchains and $K=16$ naive superchains.
    The results are computed across 10 seeds.
    The solid lines show the average MSE, and the dotted lines represent the best and worst runs across 10 seeds.
    }}
    \label{fig:MSE}
\end{figure}

In practice, we cannot measure the MSE and we rely on convergence diagnostics.
Figure~\ref{fig:scatter} plots $E^2$ against $\nR$ for a constrained superchain, and we observe a clear correlation between $\nR$ and $E^2$.
After a short warmup, the chains are far from their stationary distribution: this manifests as both a large $E^2$ and a large $\nR$.
When $\nR$ is close to 1, the squared error is smaller and approaches the distribution we would expect from stationary Markov chains (Section~\ref{sec:model-error}).
For the bimodal target, the Markov chains fail to converge after a warmup of 1000 iterations and our experiment only reports observations where $E^2$ and $\nR$ are both large.

\begin{figure}
    \centerline{
    \includegraphics[width=5in]{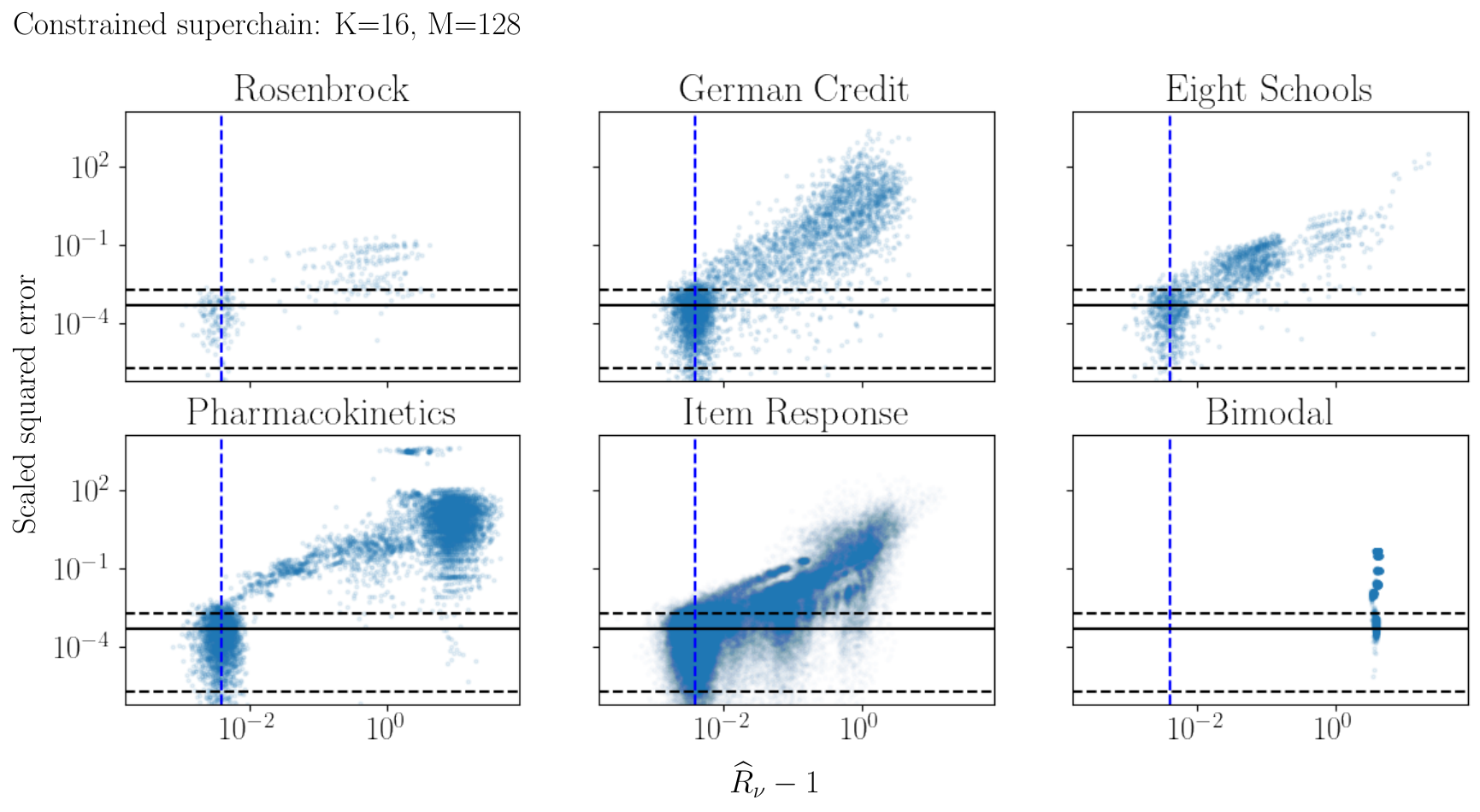}}
    \caption{\textit{Scaled squared error against $\nR$, using $K = 16$, $M = 128$ and $N = 1$, with constrained superchains.
    For short warmup phases, the Markov chains are far from their stationary distribution, which manifests as a large squared error and a large $\nR$.
    Once $\nR$ is close to 1, the scaled squared error behaves as we would expect from stationary Markov chains (Section~\ref{sec:model-error}).
    The vertical blue line corresponds to the proposed threshold $\epsilon \approx 0.004$.
    The horizontal lines show the median (solid line) and 0.9 coverage (dashed lines) for the squared error of a stationary Markov chain.}}
    \label{fig:scatter}
\end{figure}

As a tolerance on the scaled nonstationary variance, we consider $\tau = 10^{-4}$, which corresponds to a fifth of the scaled variance we tolerate with an ESS of 2000.\footnote{
  An ESS of 2000 corresponds to a relative variance, $\text{Var}_\Gamma \bar f^{(\cdot \cdot \cdot)} / \text{Var}_p f = 1 / 2000 = 5 \times 10^{-4}$,
  and the tolerance on the nonstationary variance, scaled by the estimator $\nW$ of $\text{Var}_p f$, is set to a fraction of this quantity.
}
The tolerance for $\nR$ is then $\sim$1.004, which is smaller than the recommended threshold of 1.01 for $\widehat R$.
Bear in mind the tolerance we chose is relative to our target ESS, which can change between problems, and the choice to make $\tau$ a fifth (as opposed to a half or a tenth) is arbitrary. 
Figure~\ref{fig:frac} plots the fraction of times the scaled squared error $E^2$ is above the $0.95^\text{th}$ quantile of a $\chi^2_1$ distribution for $\nR$ below varying thresholds.
For $\nR \le \sqrt{1 + 1 / M + \tau}$, this fraction approaches 0.05 for all models but can be larger. 

\begin{figure}
  \centering{
  \includegraphics[width=4.5in]{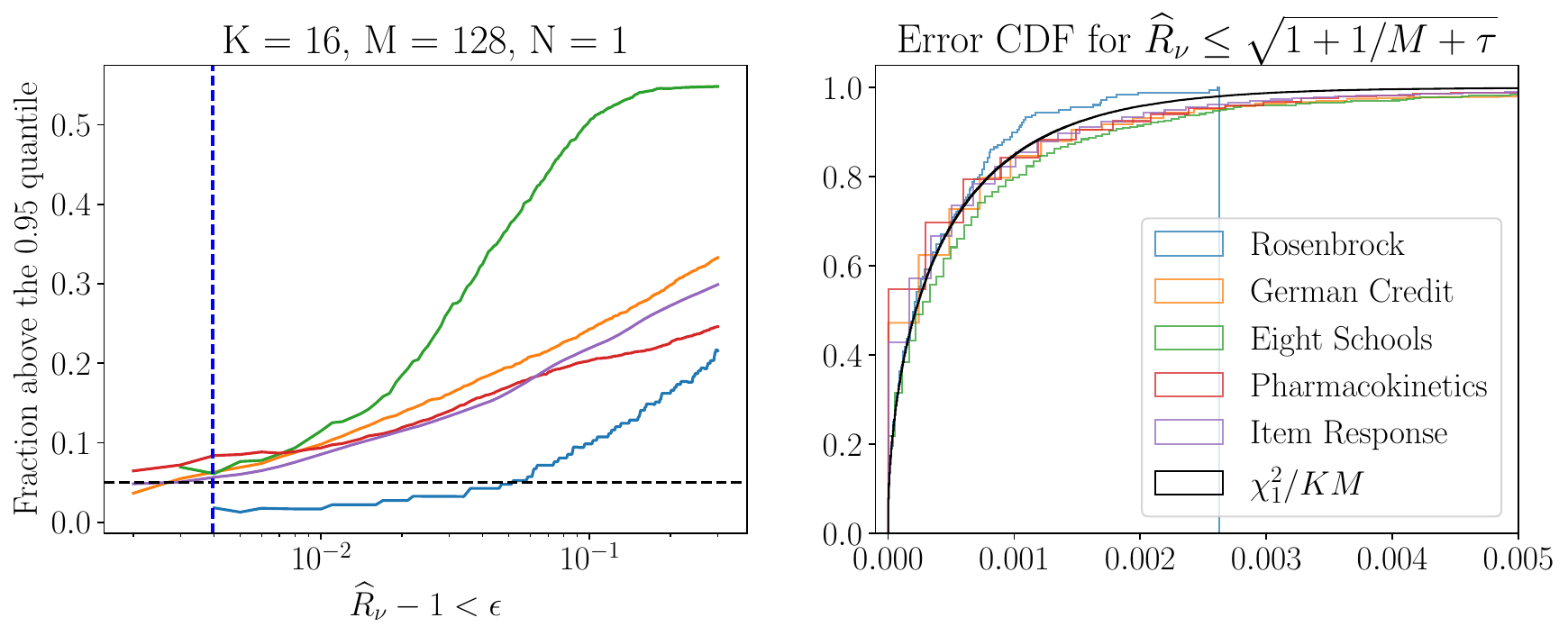}}
  \caption{\textit{(left) Fraction of Monte Carlo estimates with squared error above the $95^\text{th}$ quantile of the stationary error distribution (Section~\ref{sec:model-error}). 
  The vertical blue line is the prescribed threshold $\sqrt{1 + 1 / M + \tau}$. 
  (right) For $\nR$ close to 1, the empirical CDF approaches the theoretical CDF for stationary Markov chains.}}
  \label{fig:frac}
\end{figure}

Finally, we examine what would happen if we compute $\nR$ with naive superchains, that is without constraining all the subchains to start at the same point (Figure~\ref{fig:scatter-naive}).
In this setting, $\nR$ hardly correlates with $E^2$, and there is no threshold of $\nR$ under which the scaled squared error follows the distribution we expect from stationary Markov chains.  

\begin{figure}
    \centering
    \includegraphics[width=5in]{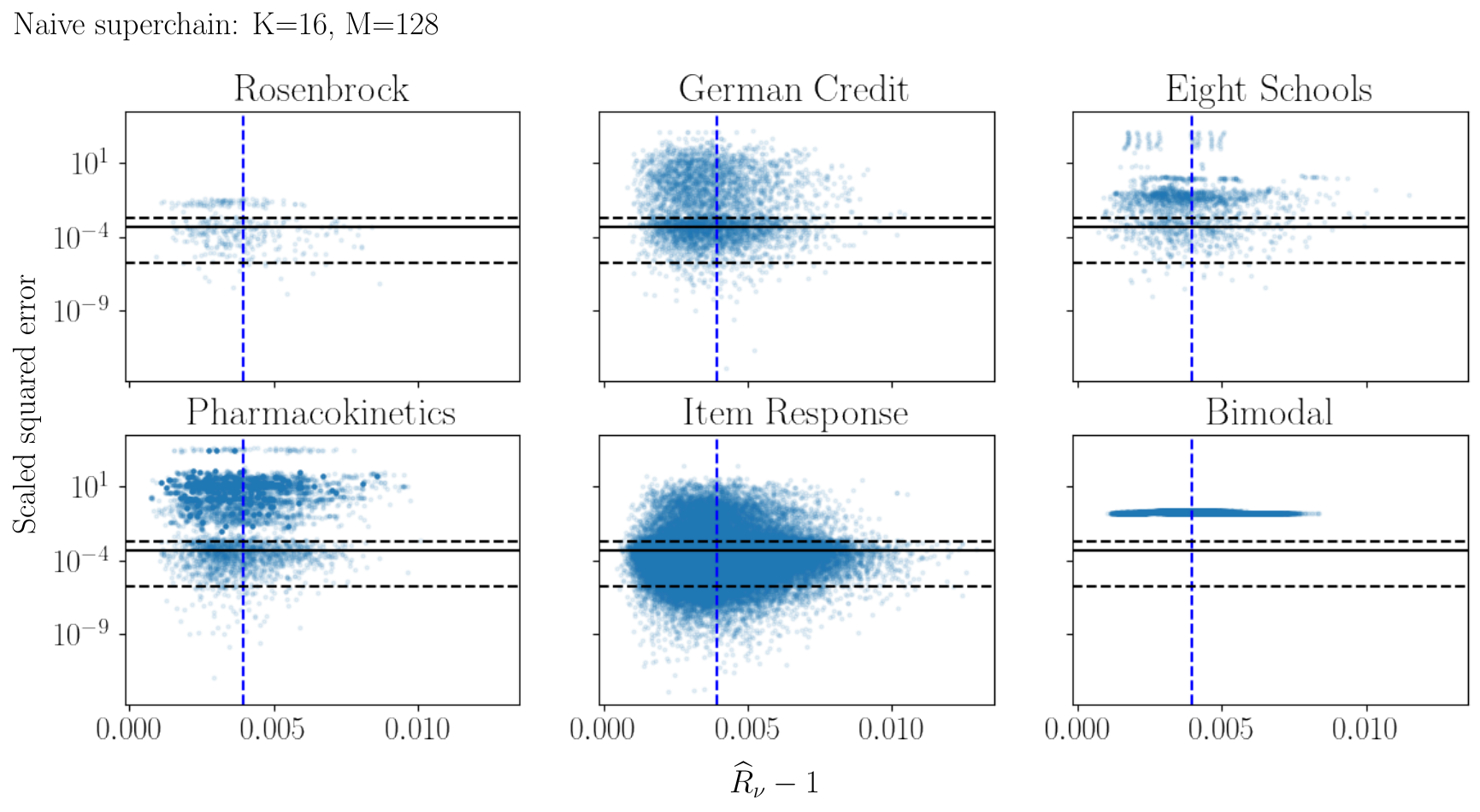}
    \caption{
    \textit{Scaled squared error against $\nR$, using $K = 16$, $M=128$, and $N = 1$, with a naive superchain. Without the constrained initialization, there is no useful correlation between $\nR$ and the squared error.}}
    \label{fig:scatter-naive}
\end{figure}

\subsection{Results when varying the number of superchains $K$}

We extend the analysis in the previous section to the case where the total number of chains stays fixed at $KM = 2048$, but we vary the number of superchains $K$.
As before, we first examine the scaled MSE, and find no clear disadvantage to using constrained superchains (Figure~\ref{fig:mse_all}).
For the Rosenbrock, German Credit, and Item Response models, increasing the number of distinct initializations improves, to varying degrees, the rate at which the MSE decreases.
For the hierarchical models (Eight Schools and Pharmacokinetics), there is no clear trend.

\begin{figure}
    \centering
    \includegraphics[width=5in]{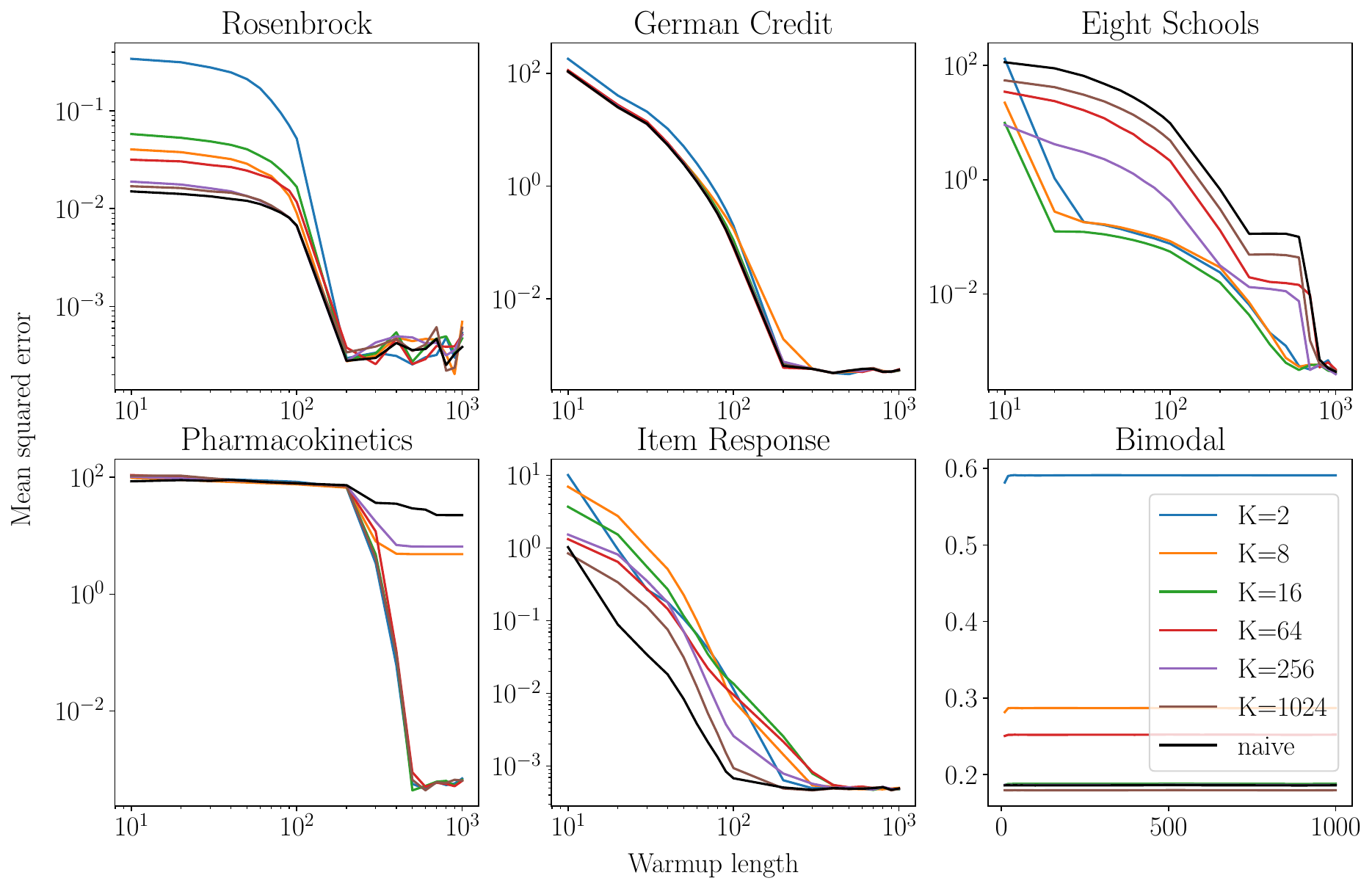}
    \caption{\textit{Mean squared error, scaled by the posterior variance, when using $K$ constrained superchains for a total of 2048 chains or using 2048 independently initialized chains (naive setting).
    The MSE is averaged across 10 seeds.}}
    \label{fig:mse_all}
\end{figure}

We next examine whether $\nR$ can diagnose convergence of the Markov chains (Figure~\ref{fig:frac_all}).
The threshold for $\nR$ is adjusted as $M$ varies, although the tolerance on the nonstationary variance $\tau$ remains fixed.
For $K \in \{8, 16, 64, 256\}$ we find that the $95^\text{th}$ quantile of $E^2$ is in reasonable agreement with the 95$^\text{th}$ quantile of a $\chi^2_1$, albeit slightly larger.
The results are less stable for the extreme choices $K = 2$ and $K = 1024$.
For $K = 2$ we expect the variance of $\nR$ to be large during the early stages of MCMC,
while for $K = 1024$, the variance is high near stationarity
(Section~\ref{sec:variance}).
Overall, there is a broad range of choices for $K$ away from these extremes that yield a functioning convergence diagnostic in the considered examples.

\begin{figure}
\centerline{  \includegraphics[width=5in]{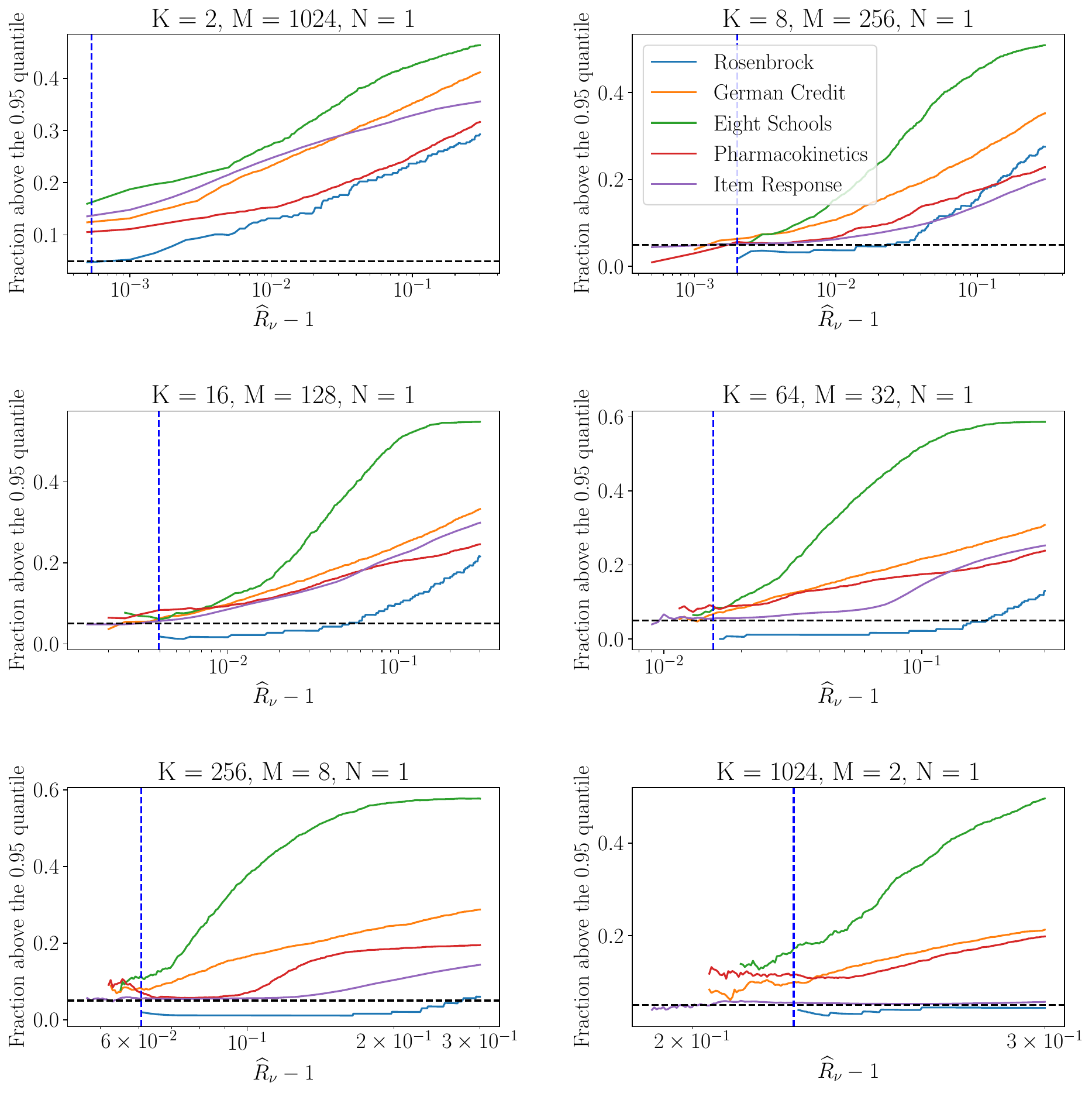}}
  \caption{\textit{Fraction of Monte Carlo estimates with squared error above the 95th quantile of the stationary error distribution (Section~\ref{sec:model-error}) when varying $K$.
  For $K$ between 8 and 256, the fraction approaches the expected 0.05 (horizontal dotted line) for stationary Markov chains, past the prescribed threshold (vertical blue line).}
  }
  \label{fig:frac_all}
\end{figure}

\section{Discussion}

While CPU clock speeds stagnate, parallel computational resources continue to get cheaper. The question of how to make effective use of these parallel resources for MCMC remains an outstanding challenge. 
This paper tackles the problem of assessing approximate convergence.

We have proposed a new convergence diagnostic, $\nR$, which is straightforward to implement for a broad class of MCMC algorithms and works for both long and short chains, in the sense that $\widehat R$ works for long chains.
A small $\nR$ (or $\widehat R$) does not guarantee convergence to the stationary distribution or that the bias has decayed to 0.
Still $\widehat R$ has empirically proven its usefulness in applied statistics, machine learning, and many scientific disciplines.
Our analysis reveals that potential success (or failure) of $\nR$ and $\widehat R$ is best understood by studying (i) the relation between the nonstationary variance and the squared bias, and (ii) how well $\nR$ monitors the nonstationary variance.

In addition to working in the many-short-chains regime, $\nR$ provides more guidance to choose a threshold, notably in the $N = 1$ case (Corollary~\ref{corr:N=1}).
%
Unlike $\widehat R$, the proposed $\nR$ requires a partition of the chains into superchains.
No choice of partition uniformly minimizes the variance of $\nR$ during all phases of MCMC. 
Based on our numerical experiments, we recommend using $K$ between 8 and 256 initializations, when running $KM = 2048$ parallel chains.
A cost of using constrained superchains is that it increases the variance of Monte Carlo estimators, although the added variance vanishes as the Markov chains converge.
A possible alternative may be to initialize all chains independently and, when constructing superchains for the computation of $\nR$, cluster chains together based on their initialization.
Studying the benefits of such an approach (both theoretically and empirically) is left as future work.


The nesting design we introduce opens the prospect of generalizing other variations on $\widehat R$, including multivariate $\widehat R$ \citep{Brooks:1998, Vats:2021, Moins:2022},
rank-normalized $\widehat R$ \citep{Vehtari:2021} and local $\widehat R$ \citep{Moins:2022}.
Nesting can further be used for less conventional convergence diagnostics, such as $R^*$, which uses classification trees to compare different chains \citep{Lambert:2022}.

A direction for future work is to adaptively set the warmup length using $\nR$.
This would follow a long tradition of using diagnostics to do early stopping of MCMC \citep{Geweke:1992, Cowles:1996, Cowles:1998, Jones:2006, Zhang:2020}.
Still, standard practice remains to prespecify the warmup length. 
This means the warmup length is rarely optimal, which is that much more exasperating in the many-short-chains regime, where the warmup dominates the computation.



\section*{Acknowledgments}

We thank the \texttt{TensorFlow Probability} team at Google, especially Alexey Radul.
We also thank Marylou Gabri\'e and Sam Livingstone for helpful discussions; Rif A. Saurous, Andrew Davison, Owen Ward, Mitzi Morris, and Lawrence Saul for helpful comments on the manuscript; and the U.S. Office of Naval Research and Research Council of Finland for partial support. 
Much of this work was done while CM was at Google Research and in the Department of Statistics at Columbia University and while LRD was in the Department of Statistics at the University of Warwick.
Finally, we thank three anonymous reviewers for their comments which have helped us improve the manuscript.
LRD was supported by the EPSRC grant EP/R034710/1. 















\appendix
\section{Proofs} \label{app:proofs}

Here we provide the proofs for formal statements throughout the paper.

\subsection{Proofs for Section~\ref{sec:nRhat}: ``Nested $\widehat R$''}

\subsubsection{Proof of Lemma~\ref{lemma:nW}: Asymptotic limit of $\nW$}

\begin{proof}
Recall the superchains are independent.
Then applying the law of large numbers along $K$ yields,
\begin{eqnarray*}
  \nW \ \xrightarrow[K \to \infty]{a.s} \ \mathbb E_\Gamma \tilde B_k + \mathbb E_\Gamma \tilde W_k.
\end{eqnarray*}
Now
\begin{eqnarray*}
  \mathbb E_\Gamma \tilde B_k & = & \frac{1}{M - 1} \sum_{m = 1}^M \mathbb E_\Gamma \left ( \bar f^{(\cdot m k)} - \bar f^{(\cdot \cdot k)} \right)^2 \\
     & = & \frac{1}{M - 1} \sum_{m = 1}^M \left ( \mathbb E_\Gamma [\bar f^{(\cdot mk)}]^2 + \mathbb E_\Gamma [\bar f^{(\cdot \cdot k)}]^2 - 
  2 \mathbb E_\Gamma (\bar f^{(\cdot m k)} \bar f^{(\cdot \cdot k)})  \right),
\end{eqnarray*}
and
\begin{eqnarray*}
  \sum_{m = 1}^M \mathbb E_\Gamma (\bar f^{(\cdot mk)} \bar f^{(\cdot \cdot k)})
    & = & M \frac{1}{M} \sum_{m = 1}^M \mathbb E_\Gamma (\bar f^{(\cdot m k)} \bar f^{(\cdot \cdot k)}) \\
    & = & M \mathbb E_\Gamma \left (\frac{1}{M} \sum_{m =1}^M \bar f^{(\cdot m k)} \bar f^{(\cdot \cdot k)} \right) \\
    &=  & M \mathbb E_\Gamma [\bar f^{(\cdot \cdot k)}]^2.
\end{eqnarray*}
Plugging this back in yields,
\begin{eqnarray*}
  \mathbb E_\Gamma \tilde B_k & = & \frac{1}{M - 1} \sum_{m = 1}^M \left ( \mathbb E_\Gamma [\bar f^{(\cdot m k)}]^2 - \mathbb E_\Gamma [\bar f^{(\cdot \cdot k)}]^2 \right)  \\
  & = & \frac{1}{M - 1} \sum_{m = 1}^N \left (\text{Var}_\Gamma \bar f^{(\cdot mk)} + [\mathbb E_\Gamma \bar f^{(\cdot m k)}]^2 - \text{Var}_\Gamma \bar f^{(\cdot \cdot k)} - [\mathbb E_\Gamma \bar f^{(\cdot \cdot k)}]^2 \right)  \\
  & = & \frac{1}{M - 1} \sum_{m = 1}^N \left [\text{Var}_\Gamma \bar f^{(\cdot m k)} - \text{Var}_\Gamma \bar f^{(\cdot \cdot k)} \right] + \left [(\mathbb E_\Gamma \bar f^{(\cdot m k)})^2 - (\mathbb E_\Gamma \bar f^{(\cdot \cdot k)})^2 \right].
\end{eqnarray*}
The second term vanishes, because $\mathbb E_\Gamma \bar f^{(\cdot m k)} = \mathbb E_\Gamma \bar f^{(\cdot \cdot k)}$.
Then
\begin{eqnarray*}
  \mathbb E \tilde B_k & = & \frac{1}{M - 1} \sum_{m = 1}^M \text{Var}_\Gamma \bar f^{(\cdot mk)} - \text{Var}_\Gamma \bar f^{(\cdot \cdot k)} \\
  & = & \frac{M}{M - 1} \left (\text{Var}_\Gamma \bar f^{(\cdot mk)} - \text{Var}_\Gamma \bar f^{(\cdot \cdot k)} \right),
\end{eqnarray*}
We next apply the law of total variance:
\begin{eqnarray*}
   \text{Var}_\Gamma \bar f^{(\cdot mk)} & = & \mathbb E_{p_0} \text{Var}_\gamma (\bar f^{(\cdot mk)} \mid \theta^k_0)
    + \text{Var}_{p_0} \mathbb E_\gamma (\bar f^{(\cdot mk)} \mid \theta^k_0),
\end{eqnarray*}
and
\begin{eqnarray*}
   \text{Var}_\Gamma \bar f^{(\cdot \cdot k)} & = & \mathbb E_{p_0} \text{Var}_\gamma (\bar f^{(\cdot \cdot k)} \mid \theta^k_0)
    + \text{Var}_{p_0} \mathbb E_\gamma (\bar f^{(\cdot \cdot k)} \mid \theta^k_0) \\
    & = & \frac{1}{M} \mathbb E_{p_0} \text{Var}_\gamma (\bar f^{(\cdot m k)} \mid \theta^k_0)
    + \text{Var}_{p_0} \mathbb E_\gamma (\bar f^{(\cdot m k)} \mid \theta^k_0),
\end{eqnarray*}
where the second line follows from noting that, conditional on $\theta_0^k$, the chains are independent,
and that $\mathbb E_\gamma (\bar f^{(\cdot \cdot k)} \mid \theta^k_0) = \mathbb E_\gamma (\bar f^{(\cdot m k)} \mid \theta^k_0)$.
Plugging this result back, we get
\begin{eqnarray*}
  \mathbb E \tilde B_k & = &  \frac{M}{M - 1} \left [ \mathbb E_{p_0} \text{Var}_\gamma (\bar f^{(\cdot m k)} \mid \theta^k_0)
   - \frac{1}{M} \mathbb E_{p_0} \text{Var}_\gamma (\bar f^{(\cdot m k)} \mid \theta^k_0) \right] \\
   & = & \mathbb E_{p_0} \text{Var}_\gamma (\bar f^{(\cdot m k)} \mid \theta^k_0).
\end{eqnarray*}
Next, if $N = 1$, $\tilde W  = 0$.
If $N > 1$, we obtain, following a similar approach as above,
\begin{eqnarray*}
  \tilde W_k & = & \frac{1}{M} \sum_{m = 1}^M \frac{1}{N - 1} \sum_{n = 1}^N \left (f^{(nmk)} - \bar f^{(\cdot mk)} \right )^2 \\
  & = & \frac{1}{M} \sum_{m = 1}^M \frac{1}{N - 1} \sum_{n = 1}^N \left (f^{(nmk)} \right )^2 - \left (\bar f^{(\cdot mk)} \right)^2.
\end{eqnarray*}
Then, taking expectations and expanding the square inside the expectation,
\begin{equation}  \label{eq:W'}
  \mathbb E_\Gamma \tilde W_k = \begin{cases}
    \frac{1}{N - 1} \sum_{n = 1}^N \text{Var}_\Gamma f^{(nmk)} - \text{Var}_\Gamma \bar f^{(\cdot mk)}
       + (\mathbb E_\Gamma f^{(nmk)})^2 - (\mathbb E_\Gamma f^{(\cdot mk)})^2 & \ \text{if} \ N > 1 \\
       0 & \ \text{if} \ N = 1,
  \end{cases}
\end{equation}
with the right side corresponding to our definition of $W'$.
Thus
  \begin{equation*}
    \W = \mathbb E_{p_0} \text{Var}_\gamma (\bar f^{(\cdot mk)} \mid \theta^k_0)+ W',
  \end{equation*}
as desired.
\end{proof}

\begin{remark}
  The second term on the right side of \cref{eq:W'} is a drift term: the (expected) sample variance increases because the samples do not have the same mean.
\end{remark}

\subsubsection{Proof of Corollary~\ref{thm:lower-nR}: stationary lower bound for $\nR$}

\begin{proof} 
For independent chains drawn at stationarity, we have for any $n$ that $\text{Var} f^{(nmk)} = \sigma^2$, where the constant $\sigma^2$ is the variance of the stationary distribution, $p$.
Furthermore, due the chain's positive autocorrelation,
\begin{equation}
  \text{Var}_\Gamma \bar f^{(\cdot mk)} \ge \frac{\sigma^2}{N}.
\end{equation}
Thus
\begin{equation*}
  \W \le \frac{1}{N - 1} \text{Var}_\Gamma \bar f^{(\cdot m k)} + \frac{N}{N - 1} \sigma^2.
\end{equation*}
Then
\begin{equation*}
  \frac{\B}{\W} \ge \frac{\text{Var}_\Gamma \bar \theta^{(\cdot mk)}}{M \left (\frac{1}{N - 1} \text{Var}_\Gamma \bar \theta^{(\cdot m k)} + \frac{N}{N - 1} \sigma^2 \right)}.
\end{equation*}
Noting that for stationary chains $\text{ESS}_{(1)} = \sigma^2 / \text{Var}_\Gamma \bar \theta^{(\cdot mk)}$,
\begin{equation*}
  \frac{\B}{\W} \ge \frac{1}{M \left ( \frac{1 + N \mathrm{ESS}_{(1)}}{N - 1} \right)} = \frac{1}{M} \frac{1 - 1 / N}{\text{ESS}_{(1)} + 1 / N}.
\end{equation*}
\end{proof}

\subsubsection{Proof of Corollary~\ref{corr:N=1}: correction for persistent variance when $N = 1$}

The proof of Corollary~\ref{corr:N=1} follows from Theorem~\ref{thm:nB} and Lemma~\ref{lemma:nW}. \\

\begin{proof}
When $N = 1$, $\tilde W_k = 0$.
Thus $\W = \mathbb E_{p_0} \text{Var}_\gamma (\bar f^{(\cdot mk)} \mid \theta^k_0)$ and
\begin{eqnarray*}
  \frac{\B}{\W} & = & \frac{\mathbb E_{p_0} \text{Var}_\gamma (\bar f^{(\cdot mk)} \mid \theta^k_0)}{M \mathbb E_{p_0} \text{Var}_\gamma (\bar f^{(\cdot mk)} \mid \theta^k_0)} +
  \frac{\text{Var}_{p_0} \mathbb E_\gamma (\bar f^{(\cdot mk)} \mid \theta^k_0)}{\mathbb E_{p_0} \text{Var}_\gamma (\bar f^{(\cdot mk)} \mid \theta^k_0)}  \\
  & = & \frac{1}{M} + \frac{\text{Var}_{p_0} \mathbb E_\gamma (\bar f^{(\cdot mk)} \mid \theta^k_0)}{\mathbb E_{p_0} \text{Var}_\gamma (\bar f^{(\cdot mk)} \mid \theta^k_0)}.
\end{eqnarray*}
Finally, $\bar f^{(\cdot mk)} = f^{(1mk)}$, given that $N = 1$.
\end{proof}


\subsubsection{Proof of Theorem~\ref{lemma:langevin}: bias and $\B / \W$ in continuous limit}

We prove Theorem~\ref{lemma:langevin} which gives us an exact expression for the bias and the ratio $\B / \W$.
The Monte Carlo estimator $\bar X_T$ is the average of $M$ diffusion processes evaluated at time $T$.
The processes are initialized at the same point $X_0 \sim p_0$ but then run independently, according to the Langevin diffusion process targeting $p = \text{normal}(\mu, \sigma)$. \\

\begin{proof}
Let $\Psi$ denote the stochastic process which generates $X_T$, and further break this process into (i) $p_0$, the process which draws $x_0$ and (ii) $\psi$ the process which generates $X_T$ conditional on $x_0$. 
Following the same arguments as in the previous sections of the Appendix, we leverage Corollary~\ref{corr:N=1},
\begin{equation*}
  \frac{\B}{\W} = \frac{1}{M} + \frac{\mathrm{Var}_{p_0} \mathbb E_\psi (X_T \mid X_0)}{\mathbb E_{p_0} \mathrm{Var}_\psi(X_T \mid X_0)}.
\end{equation*}
It is well known that Ornstein-Uhlenbeck SDEs, such as the Langevin diffusion SDE targeting a Gaussian, admit explicit solutions; see \citet{Gardiner:2003}. The solution of \cref{eq:sde} given $X_0=x_0$ is given for $T>0$ by
\begin{equation}\label{eq:langevin_solution}
X_T=e^{-T}x_0+\mu(1-e^{-T})+\sqrt{2\sigma}\int_0^Te^{-(T-s)}\dd W_s.
\end{equation}
Now, integrating with respect to $X_0\sim p_0=\text{normal}(\mu_0, \sigma_0)$ yields
\begin{equation*}
  \mathbb E_{p_0} \mathbb E_\psi (X_T \mid X_0) = \mu_0 e^{-T} + \mu(1 - e^{-T}),
\end{equation*}
then
\begin{equation*}
  \mathrm{Var}_{p_0} \mathbb E_\psi (X_T \mid X_0) = \mathrm{Var}_{p_0} \left (X_0 e^{-T} + \mu(1 - e^{-T}) \right)
  = e^{-2 T} \sigma_0^2,
\end{equation*}
and
\begin{equation*}
\mathbb E_{p_0} \mathrm{Var}_\psi (X_T \mid X_0) \\
  = \mathbb E_{p_0} \sigma^2 (1 - e^{- 2T})
  = \sigma^2 (1 - e^{-2T}),
\end{equation*}
from which the desired result follows.
\end{proof}

\subsubsection{Proof of Corollary~\ref{thm:langevin}: initialization conditions under which $\nR$ is reliable}

Theorem~\ref{thm:langevin} provides conditions in the continuous limit under which $\nR$ is $(\delta, \delta')$-reliable and formalizes the notion of overdispersed intializations. \\

\begin{proof}
  The bias is given by
\begin{equation*}
\mathbb E (X_T) - \mathbb E_{p} (\mu) = (\mu_0 - \mu) e^{-T},
\end{equation*}
which is a monotone decreasing function of $T$.
The time at which the scaled squared bias is below $\delta'$ is obtained by solving
\begin{eqnarray*}
  \frac{(\mu_0 - \mu)^2 e^{- 2T}}{\sigma^2} \le \delta'.
\end{eqnarray*}
If $(\mu_0 - \mu)^2 / \sigma^2 \le \delta'$, then the above condition is verified for any $T$ and $\nR$ is trivially $(\delta, \delta')$-reliable.
Suppose now that $(\mu_0 - \mu)^2 / \sigma^2 > \delta'$.
Then we require
\begin{equation*}
  T \ge \frac{1}{2} \log \left (\frac{(\mu_0 - \mu)^2}{\delta' \sigma^2} \right) \triangleq T^*.
\end{equation*}
It remains to ensure that for $T < T^*$, $\B / \W > \delta$.
Plugging in $T^*$ in the expression from Lemma~\ref{lemma:langevin} and noting $\B / \W$ is monotone decreasing in $T$, we have
\begin{equation*}
  \sigma^2_0 > \left (\delta - \frac{1}{M} \right) \left (e^{2 T^*} - 1 \right) \sigma^2,
\end{equation*}
which is the wanted expression.
\end{proof}


\section{Additional results on the reliability of $\nR$ and $\widehat R$} \label{app:reliability}

This appendix provides additional results on the reliability of $\nR$ (Section~\ref{sec:reliable-theory}).
We numerically test the lower bound on the initial variance, provided by Corollary~\ref{thm:langevin}, on a Gaussian target and a mixture of two Gaussians.
We then derive a reliability condition for $\widehat R$, tackling the $N > 1$ case where we have more than one sample per chain.

\subsection{Numerical evaluation of the reliability of $\nR$}

We numerically evaluate the $(\delta, \delta')$-reliability of $\nR$ on two examples:
\begin{itemize}
    \item[](i) a standard Gaussian, which conforms to the assumption of our theoretical analysis.
    \item[](ii) a mixture of two Gaussians, which constitutes a canonical example where $\widehat R$ potentially fails.
    In this example the Markov chains fail to mix, hence reliability is achieved if $\nR \ge \sqrt{1 + \delta}$.
\end{itemize}

We approximate the Langevin diffusion using the Metropolis adjusted Langevin algorithm (MALA) algorithm, which is equivalent to HMC with a one-step leapfrog integrator.
The step size is 0.04, which is chosen to be as small as possible while ensuring that $\nR$ reports convergence after $\sim$$2 \times 10^4$ iterations for the standard Gaussian target.
We set $M = 16$, $\delta \approx 0.1$ and $\delta' = \delta / 5$.
Reliability is defined in terms of $\B$ and $\W$, which are the asymptotic limits of $\nB$ and $\nW$ when $K \to \infty$.
We approximate this limit by using $K = 1024$ superchains for a total of 16384 chains.

\begin{figure}
    \centering
    \includegraphics[width=4.5in]{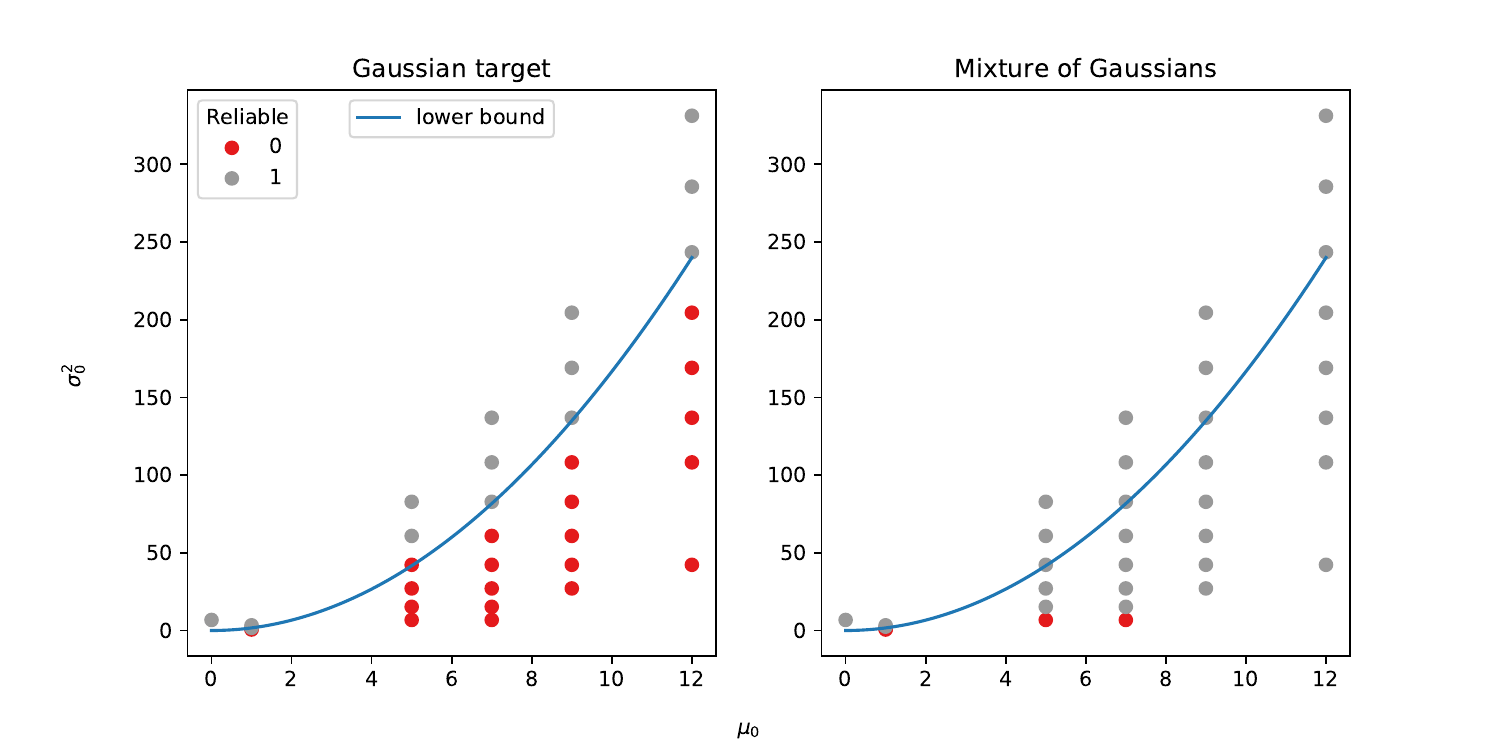}
    \caption{\em $(\delta, \delta')$-reliability of $\nR$ for varying initial bias $\mu_0$ and initial variance $\sigma_0^2$.
    }
    \label{fig:reliable}
\end{figure}

The results are shown in Figure~\ref{fig:reliable}.
The theoretical lower bound (Theorem~\ref{thm:langevin}) is accurate when using a standard Gaussian target.
When targeting a mixture of Gaussians, the lower bound is too conservative and $\nR$ is reliable even when using an ``underdispersed'' initialization.
This is because we use a large number of distinct initializations and, even with a small initial variance, we typically find both modes and identify poor mixing.
This suggests the failure of $\widehat R$ on multimodal targets is often due to using too few distinct initializations, rather than an underdispersed initialization.

\subsection{Reliability condition for $\widehat R$ in the continuous limit}
\label{app:N>1}

We here conduct an analysis for $\widehat R$ similar to the one conducted for $\nR$ in Section~\ref{sec:reliable-theory}.
That is we examine the reliability of $\widehat R$ when approximating the MCMC chain by a Langevin diffusion which targets a Gaussian.

There are three differences when compared to our study of $\nR$: (i) each Monte Carlo estimator is now made up of only one chain and all chains are independent, (ii) the chains do not include warmup, i.e., $\mathcal{W}=0$, and (iii) the length of each chain is chosen as $N =\lfloor T/h\rfloor$ for a small value step size $h$ such that for each chain $m$, the distribution of $\theta^{(\lfloor t/h\rfloor m)}$ can be approximated by the Langevin solution defined in \eqref{eq:langevin_solution}.
The distribution of each (within chain) estimator can therefore be approximated by the distribution of
\begin{equation}
    \overline{X}_T = \frac{1}{T} \int_0^T X_s \text ds.
\end{equation}
In this framework, the limits of $B$ and $W$ as $h\rightarrow0$ yield
$$
 B  = \mathrm{Var}_\Psi (\bar X_T),\qquad W =  \frac{1}{T} \int_0^T \mathbb{E}\left[\left (X_t - \bar X_T \right)^2\right].
$$
This next lemma provides an exact expression for $B / W$.
\begin{lemma} \label{lemma:langevin-rhat}
  Suppose we initialize a process at $X_0 \sim p_0$, which evolves according to \eqref{eq:sde} from time $t = 0$ to $t = T > 0$, and let $\bar X_T$ be defined as above.
  Denote its distribution as $\Psi$.
  Let
  \begin{eqnarray*}
      \rho_T & \triangleq & \frac{1}{T} (1 - e^{-T}), \\
      \xi_T & \triangleq &  \frac{1}{2T} \left(1 -  e^{-2T} \right).\\
      \eta_T&\triangleq&\frac{2}{T}(1-\rho_T).
  \end{eqnarray*}
  Then
  \begin{equation}
      \mathbb E_\Psi \bar X_T - \mathbb E_p X = (\mu_0 - \mu) \rho_T,
  \end{equation}
  and
  \begin{equation}
      \frac{B}{W} = \frac{(\sigma_0^2 - \sigma^2) \rho^2_T
        + \sigma^2\eta_T}
        {(\sigma^2_0 - \sigma^2 + (\mu_0-\mu)^2)(\xi_T - \rho_T^2) +\sigma^2 (1 - \eta_T) }.
  \end{equation}
\end{lemma}
\ \\

\begin{proof}
  Let $\Psi$ denote the stochastic process which generates $X_t$. 
  Once again, we exploit the explicit solution \eqref{eq:langevin_solution} to the Langevin SDE.
  We begin with the numerator.
  \begin{align*}
      B & = 
      \mathrm{Var}_\Psi (\bar X_T) \\ 
      &=\frac{1}{T^2}\int_0^T\int_0^T{\rm Cov}_\Psi (X_s, X_t)\dd s \dd t\\
    &=\frac{1}{T^2}\int_0^T\int_0^T\left((\sigma_0^2-\sigma^2)e^{-(s+t)}+\sigma^2e^{-|s-t|}\right)\dd s\dd t\\
    &=(\sigma_0^2-\sigma^2)\rho^2_T+2\left(\frac{\sigma^2}{T}+(-1+e^{-T})\frac{\sigma^2}{T^2}\right) \\
    &=(\sigma^2_0 - \sigma^2) \rho^2_T + \sigma^2 \eta_T.
  \end{align*}
Next we have
\begin{equation*}
    W =  \frac{1}{T} \int_0^T \mathbb{E}\left[\left (X_t - \bar X_T \right)^2\right].
\end{equation*}
Following the same steps to prove Lemma~\ref{lemma:nW}, we have
\begin{equation*}
     W = \frac{1}{T} \int_0^T \ \left [ \mathrm{Var}_\Psi (X_t) - \mathrm{Var}_\Psi (\bar X_T) \right] + \left[ (\mathbb E_\Psi (X_t))^2 - (\mathbb E_\Psi (\bar X_T))^2 \right] \dd t.
\end{equation*}
Computing each term yields
  \begin{align*}
             \mathbb E_\Psi (X_t)&=\mu+(\mu_0-\mu)e^{-t}\\
             \mathrm{Var}_\Psi(X_t)&=\sigma^2+(\sigma_0^2-\sigma^2)e^{-2t}\\
             \mathbb E_\Psi(\bar X_T)&=\mu+(\mu_0-\mu)\rho_T\\
             \mathrm{Var}_\Psi(\bar X_T)&=\sigma^2\eta_T+(\sigma_0^2-\sigma^2)\rho_T^2
         \end{align*}

Constructing $W$ term by term,
\begin{align*}
    \frac{1}{T} \int_0^T  (\mathrm{Var}_\Psi (X_t)-\mathrm{Var}_\Psi(\bar X_T)) \dd t 
    & = \frac{1}{T} \int_0^T  \sigma^2(1-\eta_T) + (\sigma^2-\sigma_0^2) ( e^{-2t}-\rho_T^2) \dd t\\
     &= \sigma^2(1-\eta_T) +  (\sigma^2_0 - \sigma^2)(\xi_T-\rho_T^2),
\end{align*}
Similarly,
\begin{align*}
        \frac{1}{T} \int_0^T  ((\mathbb E_\Psi (X_t))^2-(\mathbb E_\Psi (\bar X_T))^2) \dd t&= \frac{1}{T} \int_0^T  (\mu_0-\mu) (e^{-t}-\rho_T)(2\mu+(\mu_0-\mu)(e^{-t}+\rho_T)) \dd t\\
    &= (\mu_0-\mu)^2(\xi_T-\rho_T^2).
\end{align*}
Putting it all together, we have
\begin{equation*}
    W = (\sigma^2_0 - \sigma^2 + (\mu_0-\mu)^2)(\xi_T - \rho_T^2) +\sigma^2 (1 - \eta_T) .
\end{equation*}
\end{proof}

\begin{remark}
  Taking the limit at $T \to 0$ yields
  $$
  \underset{T \to 0}{\lim} \ \rho_T = \underset{T \to 0}{\lim} \ \xi_T=\underset{T \to 0}{\lim} \ \eta_T= 1
  $$
  Thus $\underset{T \to 0}{\lim} \ B  = \sigma_0^2$ and $\underset{T \to 0}{\lim} \  W = 0$, therefore $\underset{T \to 0}{\lim} \ B / W = +\infty$.
\end{remark}
The above limits can be calculated by Taylor expanding the exponential.

We now state the main result of this section, which provides a lower bound on $\sigma_0^2$ in order to insure $\widehat R$ is $(\delta, \delta')$-reliable.
Unlike in the $\nR$ case the proof requires some additional assumptions.

\begin{theorem} \label{thm:rhat_reliable}
  If $(\mu - \mu_0)^2 / \sigma^2 \le \delta'$, then $\widehat R$ is always $(\delta, \delta')$-reliable.
  Suppose now that $(\mu - \mu_0)^2 / \sigma^2 > \delta'$.
  Let $T^*$ solve
  \begin{equation*}
      \frac{(\mu - \mu_0)^2 \rho^2_T}{\sigma^2} = \delta',
  \end{equation*}
  for $T$.
  Assume:
  \begin{itemize}
      \item[] (A1) $B / W$ is monotone decreasing (conjecture: this is always true).
      \item[] (A2) $\delta$ verifies the upper bound \begin{equation*} \delta < \frac{1}{\frac{1}{2} T^* \coth \left (\frac{T^*}{2} \right) - 1}
      \end{equation*}
      where $\coth$ is the hyperbolic cotangent.
  \end{itemize}
  Then $\widehat R$ is $(\delta, \delta')$-reliable if and only if
  \begin{equation} \label{eq:rhat_bound}
      \sigma_0^2 \ge \frac{\delta (\xi_{T^*} - \rho^2_{T^*}) (\mu - \mu_0)^2 + [\delta(1 + \rho^2_{T^*} - \eta_{T^*} - \xi_{T^*}) - (\eta_{T^*} - \rho^2_{T^*})] \sigma^2}{(1 + \delta) \rho^2_{T^*} - \delta \xi_{T^*}}.
  \end{equation}
\end{theorem}
\ \\

\begin{proof}
  When $(\mu - \mu_0)^2 / \sigma^2 \le \delta'$, $(\delta, \delta')$-reliability follows from the fact $\rho_T$ and thence the bias are monotone decreasing.

  Consider now the case where $(\mu - \mu_0)^2 / \sigma^2 > \delta'$.
  To alleviate the notation, assume without loss of generality that $\mu = 0$.
  Per Assumption (A1), it suffices to check that for $T = T^*$, $B / W \ge \delta$.
  Per Lemma~\ref{lemma:langevin-rhat}, this is equivalent to
  \begin{eqnarray*}
  & \frac{(\sigma_0^2 - \sigma^2) \rho^2_{T^*}
        + \sigma^2\eta_{T^*}}
        {(\sigma^2_0 - \sigma^2 + \mu_0^2)(\xi_{T^*} - \rho_{T^*}^2) +\sigma^2 (1 - \eta_{T^*}) } & \ge \delta \\[0.25in]
  \iff & \frac{\sigma_0^2 \rho^2_{T^*} + (\eta_{T^*} - \rho_{T^*}^2) \sigma^2}{\sigma^2_0(\xi_{T^*} - \rho_{T^*}^2) + (\mu_0^2 - \sigma^2) (\xi_{T^*} - \rho_{T^*}^2) +
  \sigma^2 (1 - \eta_{T^*})} & \ge \delta \\[0.25in]
%
%
  \iff & \sigma^2_0 (\rho^2_{T^*} + \delta (\rho_{T^*}^2 - \xi_{T^*})) 
  & \ge \delta \left[ (\mu_0^2 - \sigma^2) (\xi_{T^*} - \rho_{T^*}^2) \right] + \delta (1 - \eta_{T^*}) \sigma^2 \\
  & & \hspace{0.5in} - (\eta_{T^*} - \rho_{T^*}^2) \sigma^2  \\[0.25in]
  \iff & \sigma^2_0 [(1 + \delta) \rho^2_{T^*} - \delta \xi_{T^*}] 
  & \ge \delta \mu_0^2 (\xi_{T^*} - \rho_{T^*}^2) \\ 
  & & \hspace{0.15in} + [\delta(1 + \rho_{T^*}^2 - \eta_{T^*}- \xi_{T^*}) - (\eta_{T^*} - \rho_{T^*}^2)] \sigma^2.
  \end{eqnarray*}
To complete the proof, we need to show that $\left ( (1 + \delta) \rho^2_{T^*} - \delta \xi_{T^*} \right)$ is positive.
This will not always be true, hence the requirement for Assumption (A2).
We arrive at this condition by expressing $\xi_T$ in terms of $\rho^2_T$.
\begin{eqnarray*}
  \xi_T & = & \frac{\xi_T}{\rho_T} \rho_T \\
  & = & \frac{1}{2} \left ( \frac{1 - e^{- 2 T}}{1 - e^{-T}} \right) \rho_T \\
  & = & \frac{1}{2} \left (\frac{1 - e^{-T} + e^{-T} - e^{-2T}}{1 - e^{-T}} \right) \rho_T \\
  & = & \frac{1}{2} \left (1 + e^{-T} \frac{1 - e^{-t}}{1 - e^{-T}} \right) \rho_T \\
  & = & \frac{1}{2} (1 + e^{-T}) \rho_T  \\
  & = & \frac{1}{2} \left(\frac{1 + e^{-T}}{\rho_T}\right) \rho_T^2 \\
  & = & \frac{1}{2} \left(T \frac{1 + e^{-T}}{1 - e^{-T}}\right) \rho^2_T \\
  & = & \frac{1}{2} T \coth(T / 2) \rho^2_T.
\end{eqnarray*}

Thus
\begin{eqnarray*}
   (1 + \delta) \rho^2_{T^*} - \delta \xi_{T^*} & = & \rho^2_{T^*} \left [1 + \delta - \frac{\delta}{2} T \coth(T / 2) \right],
\end{eqnarray*}
which by assumption (A2) is positive.
\end{proof}

\begin{remark}
On the right side of \eqref{eq:rhat_bound}, all terms in parenthesis in the numerator are positive, meaning the numerator comprises a positive term scaled by $\delta$,
\begin{equation*}
    \delta \left [(\xi_{T^*} - \rho^2_{T^*}) (\mu - \mu_0)^2 + (1 + \rho^2_{T^*} - \eta_{T^*} - \xi_{T^*}) \sigma^2 \right],
\end{equation*}
and a negative term,
\begin{equation*}
    - (\eta_{T^*} - \rho_{T^*}^2) \sigma^2.
\end{equation*}
This second term appears in the expression for $B$ (Lemma~\ref{lemma:langevin-rhat}), which we can rewrite as
\begin{equation*}
    B = \sigma^2_0 \rho^2_T + \sigma^2 (\eta_T - \rho^2_T) \ge \sigma^2 (\eta_T - \rho^2_T).
\end{equation*}
This lower bound does not cancel with $W$, ensuring that $B / W$ is nonzero.
Hence for 
\begin{equation*}
    \delta \le \frac{(\eta_{T^*} - \rho^2_{T^*}) \sigma^2}{(\xi_{T^*} - \rho^2_{T^*}) (\mu - \mu_0)^2 + (1 + \rho^2_{T^*} - \eta_{T^*} - \xi_{T^*}) \sigma^2}
\end{equation*}
the reliability condition is  always met, including even when $\sigma_0^2 = 0$.
This is comparable to the $\delta < 1 / M$ case in Theorem~\ref{thm:langevin}.
\end{remark}

We can understand Assumption (A2) as a requirement that $\delta$ be not too large compared to $\delta'$, since a smaller $\delta'$ implies a larger $T^*$. \textit{Why} such a requirement exists remains conceptually unclear.
We simulate the upper bound for $\delta' \in (0, 1)$, $\sigma = 1$ and $|\mu_0 - \mu| \in \{1, 2, \cdots, 5\}$ (Figure~\ref{fig:delta_bound}).
In the studied cases, the upper bound on $\delta$ is always at least $\sim$3 times $\delta'$ and potentially orders of magnitude larger.

\begin{figure}
    \centering
    \includegraphics[width=5in]{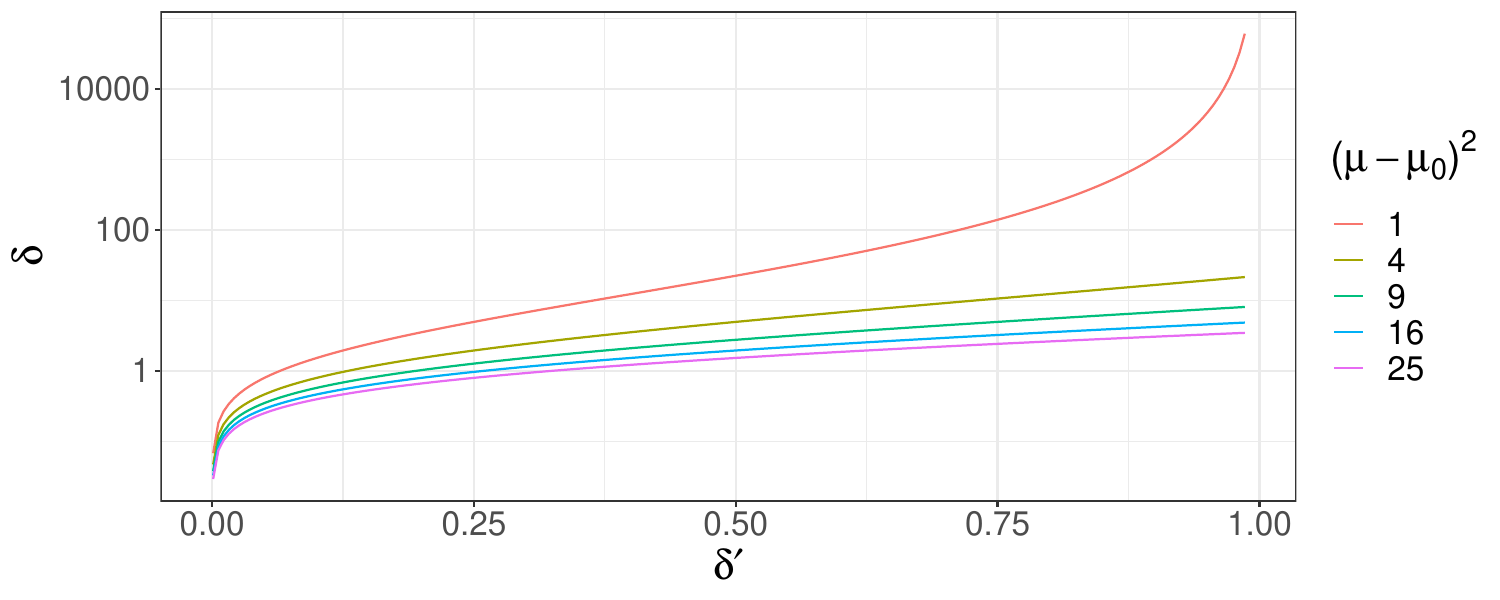}
    \caption{\em Upper bound on $\delta$ to verify Assumption (A2) in Theorem~\ref{thm:rhat_reliable}}
    \label{fig:delta_bound}
\end{figure}

%
%
%

\section{Description of models} \label{app:models}

We provide details on the target distributions used in Section~\ref{sec:experiments}.
For the Rosenbrock, German Credit, Eight Schools, and Item Response models, precise estimates of the mean and variance along each dimension can be found in the Inference Gym \citep{inferencegym:2020}.
For the Pharmacokinetics example, we compute benchmark means and variances using 2048 chains, each with 1000 warmup and 1000 sampling iterations.
We run MCMC using TensorFlow Probability's implementation of ChEES-HMC \citep{Hoffman:2021}.
The resulting effective sample size is between 60,000 and 100,000 depending on the parameters.
For the bimodal target, the correct mean and variance are worked out analytically.

\subsection{Rosenbrock (Dimension = 2)}

The Rosenbrock distribution is a nonlinearly transformed normal distribution with highly non-convex level sets; see Equation (\ref{eq:rosenbrock}).

\subsection{German Credit (Dimension = 25)}

``German Credit'' is a Bayesian logistic regression model applied to a dataset from a machine learning repository \citep{Dua:2017}.
There are 24 features and an intercept term.
The joint distribution over $(\theta, y)$  is
\begin{eqnarray*}
    \theta & \sim & \text{normal}(0, I), \\
    y_n & \sim & \text{Bernoulli} ( \mbox{logit}^{-1}(\theta^T x_n)),
\end{eqnarray*}
where $I \in \mathbb R^{24 \times 24}$ is the identity matrix.
Our goal is to sample from the posterior distribution $p(\theta \mid y)$.

\subsection{Eight Schools (Dimension = 10)} 
``Eight Schools'' is a Bayesian hierarchical model describing the effect of a program to train students to perform better on a standardized test, as measured by performance across 8 schools \citep{Rubin:1981}.
We estimate the group mean and the population mean and variance.
To avoid a funnel shaped posterior density, we use a non-centered parameterization:
\begin{eqnarray*}
    \mu & \sim &\text{normal}(5, 3) \\ 
    \sigma &\sim & \text{normal}^+(0, 10) \\
    \eta_n & \sim & \text{normal}(0, 1) \\
    \theta_n & = & \mu + \eta_n \sigma \\
    y_n & \sim & \text{normal}(\theta_n, \sigma_n),
\end{eqnarray*}
with the posterior distribution taken over $\mu$ and $\eta$.

\subsection{Pharmacokinetic (Dimension = 45)} 

``Pharmacokinetics'' is a one-compartment model with first-order absorption from the gut that describes the diffusion of a drug compound inside a patient's body.
Oral administration of a bolus drug dose induces a discrete change in the drug mass inside the patient's gut.
The drug is then absorbed into the \textit{central compartment}, which represents the blood and organs into which the drug diffuses profusely.
This diffusion process is described by the system of differential equations:
\begin{eqnarray*}
  \frac{\mathrm d m_\text{gut}}{\mathrm d t} & = & - k_1 m_\text{gut} \nonumber  \\
  \frac{\mathrm d m_\text{cent}}{\mathrm d t} & = & k_1 m_\text{gut} - k_2 m_\text{cent},
\end{eqnarray*}
which admits the analytical solution, when $k_1 \neq k_2$,
\begin{eqnarray*}
  m_\text{gut}(t) & = & m^0_\text{gut} \exp(- k_1 t) \nonumber \\
  m_\text{cent}(t) & = & \frac{\exp(- k_2 t)}{k_1 - k_2} \left (m^0_\text{gut} k_1 (1 - \exp[(k_2 - k_1) t] + (k_1 - k_2) m^0_\text{cent}) \right).
\end{eqnarray*}
Here $m^0_\text{gut}$ and $m^0_\text{cent}$ are the initial conditions at time $t = 0$.

A patient typically receives multiple doses.
To model this, we solve the differential equations between dosing events, and then update the drug mass in each compartment, essentially resetting the boundary conditions before we resume solving the differential equations.
In our example, this means adding $m_\text{dose}$, the drug mass administered by each dose, to $m_\text{gut}(t)$ at the time of the dosing event.
We label the dosing schedule as $x$.

Each patient receives a total of 3 doses, taken every 12 hours.
Measurements are taken at times \mbox{$t = (0.083, 0.167, 0.25, 0.5, 0.75, 1, 1.5, 2, 3, 4, 6, 8)$} hours after each dosing event. 

We simulate data for 20 patients and for each patient, indexed by $n$, we estimate the coefficients $(k_1^n, k_2^n)$.
We use a hierarchical prior to pool information between patients and estimate the population parameters $(k_1^\text{pop}, k_2^\text{pop})$ with a non-centered parameterization.
The full Bayesian model is:
\begin{eqnarray*}
  \text{\textit{hyperpriors:}} \\
  k_1^\text{pop} & \sim & \text{lognormal}(\log 1, 0.1) \\
  k_2^\text{pop} & \sim & \text{lognormal}(\log 0.3, 0.1) \\
  \sigma_1 & \sim & \text{lognormal}(\log 0.15, 0.1) \\
  \sigma_2 & \sim & \text{lognormal}(\log 0.35, 0.1) \\
  \sigma & \sim & \text{lognormal}(-1, 1) \\
  \text{\textit{hierarchical \ priors:}} \\
  \eta_1^n & \sim & \text{normal}(0, 1) \\
  \eta_2^n & \sim & \text{normal}(0, 1) \\
  k_1^n & = & k_1^\text{pop} \exp(\eta_1^n \sigma_1) \\
  k_2^n & = & k_2^\text{pop} \exp(\eta_2^n \sigma_2) \\
  \text{\textit{likelihood:}} \\
  y_n & \sim & \text{lognormal}(\log m_\text{cent} (t, k_1^n, k_2^n, x), \sigma).
\end{eqnarray*}
We fit the model on the unconstrained scale, meaning the Markov chains explore the parameter space of, for example, $\log k_1^\text{pop} \in \mathbb R$, rather than $k_1^\text{pop} \in \mathbb R^+$.

\subsection{Mixture of Gaussians (Dimension = 100)}

A mixture of two well-separated 100-dimensional normals,
\begin{equation*}
    \theta \sim 0.3 \ \text{MVN}(-\mu, I) + 0.7 \ \text{MVN}(\mu, I),
\end{equation*}
where $\mu$ is the 100-dimensional vector of 5's.

\subsection{Item Response Theory (Dimension = 501)}

The posterior of a model to estimate students abilities when taking an exam.
There are $J = 400$ students and $L = 100$ questions.
The model parameters are the mean student ability $\delta \in \mathbb R$, the ability of each individual student $\boldsymbol \alpha \in \mathbb R^J$ and the difficulty of each question $\boldsymbol \beta \in \mathbb R^L$.
The observations are the binary matrix $Y \in \mathbb R^{J \times L}$, with $y_{j \ell}$ the response of student $j$ to question $\ell$.
The joint distribution is
\begin{eqnarray*}
  \delta & \sim & \text{normal}(0.75, 1) \\
  \boldsymbol \alpha & \sim & \text{MVN}(0, I) \\
  \boldsymbol \beta & \sim & \text{MVN}(0, I) \\
  y_{j \ell} & \sim & \text{Bernoulli} (\text{logit}^{-1} (\alpha_j - \beta_\ell + \delta)).
\end{eqnarray*}

\bibliographystyle{ba}
\bibliography{ref.bib}

\begin{thebibliography}{55}
\newcommand{\enquote}[1]{``#1''}
\expandafter\ifx\csname natexlab\endcsname\relax\def\natexlab#1{#1}\fi
\expandafter\ifx\csname url\endcsname\relax
  \def\url#1{{\tt #1}}\fi
\expandafter\ifx\csname urlprefix\endcsname\relax\def\urlprefix{URL }\fi
\ifx\endbibitem\undefined \let\endbibitem\relax\fi

\bibitem[{Andrieu and Thoms(2008)}]{Andrieu:2008}
Andrieu, C. and Thoms, J. (2008).
\newblock \enquote{A tutorial on adaptive {MCMC}.}
\newblock {\em Statistics and Computing\/}, 18: 343--376.
\endbibitem

\bibitem[{Beskos et~al.(2013)Beskos, Pillai, Roberts, Sanz-Serna, and
  Stuart}]{Beskos:2013}
Beskos, A., Pillai, N., Roberts, G., Sanz-Serna, J.-M., and Stuart, A. (2013).
\newblock \enquote{Optimal tuning of the hybrid {Monte Carlo} algorithm.}
\newblock {\em Bernoulli\/}, 19(5A): 1501--1534.
\endbibitem

\bibitem[{Betancourt(2018)}]{Betancourt:2018}
Betancourt, M. (2018).
\newblock \enquote{A conceptual introduction to {Hamiltonian Monte Carlo}.}
\newblock {\em arXiv:1701.02434v1\/}.
\endbibitem

\bibitem[{Brooks and Gelman(1998)}]{Brooks:1998}
Brooks, S.~P. and Gelman, A. (1998).
\newblock \enquote{General methods for monitoring convergence of iterative
  simulations.}
\newblock {\em Journal of Computational and Graphical Statistics\/}, 7:
  434--455.
\endbibitem

\bibitem[{Buekner et~al.(2024)Buekner, Gabry, Kay, and Vehtari}]{posterior}
Buekner, P., Gabry, J., Kay, M., and Vehtari, A. (2024).
\newblock \enquote{posterior: Tools for working with posterior distributions.}
\newline\urlprefix\url{https://github.com/stan-dev/posterior}
\endbibitem

\bibitem[{Cowles and Carlin(1996)}]{Cowles:1996}
Cowles, M.~K. and Carlin, B.~P. (1996).
\newblock \enquote{Markov chain {Monte Carlo} convergence diagnostics: A
  comparative review.}
\newblock {\em Journal of the American Statistical Association\/}, 91:
  883--904.
\endbibitem

\bibitem[{Cowles et~al.(1998)Cowles, Roberts, and Rosenthal}]{Cowles:1998}
Cowles, M.~K., Roberts, G.~O., and Rosenthal, J.~S. (1998).
\newblock \enquote{Possible biases induced by {MCMC} convergence diagnostics.}
\newblock {\em Journal of Statistical Computation and Simulation\/}, 64:
  87--104.
\endbibitem

\bibitem[{Del~Moral et~al.(2006)Del~Moral, Doucet, and Jasra}]{DelMoral:2006}
Del~Moral, P., Doucet, A., and Jasra, A. (2006).
\newblock \enquote{Sequential {Monte Carlo} samplers.}
\newblock {\em Journal of the Royal Statistical Society, Series B\/}, 68:
  411--436.
\endbibitem

\bibitem[{du~Ch\'e and Margossian(2023)}]{DuChe:2023}
du~Ch\'e, S. and Margossian, C.~C. (2023).
\newblock \enquote{Parallelization for {Markov} chain {Monte Carlo} with
  heterogeneous runtimes.}
\newblock {\em BayesComp\/}.
\endbibitem

\bibitem[{Dua and Graff(2017)}]{Dua:2017}
Dua, D. and Graff, C. (2017).
\newblock \enquote{{UCL} machine learning repository.}
\newline\urlprefix\url{http://archive.ics.ucl.edu/ml}
\endbibitem

\bibitem[{Flegal et~al.(2008)Flegal, Haran, and Jones}]{Flegal:2008}
Flegal, J.~M., Haran, M., and Jones, G.~L. (2008).
\newblock \enquote{Markov chain {Monte Carlo}: Can we trust the third
  significant figure?}
\newblock {\em Statistical Science\/}, 250--260.
\endbibitem

\bibitem[{Gardiner(2004)}]{Gardiner:2003}
Gardiner, C.~W. (2004).
\newblock {\em Handbook of Stochastic Methods for Physics, Chemistry and the
  Natural Sciences, 3rd edition\/}.
\newblock Springer-Verlag, Berlin.
\endbibitem

\bibitem[{Gelman et~al.(2013)Gelman, Carlin, Stern, Dunson, Vehtari, and
  Rubin}]{Gelman:2013}
Gelman, A., Carlin, J.~B., Stern, H.~S., Dunson, D., Vehtari, A., and Rubin,
  D.~B. (2013).
\newblock {\em Bayesian Data Analysis, 3rd edition\/}.
\newblock CRC Press.
\endbibitem

\bibitem[{Gelman et~al.(1997)Gelman, Gilks, and Roberts}]{Gelman:1997}
Gelman, A., Gilks, W.~R., and Roberts, G.~O. (1997).
\newblock \enquote{Weak convergence and optimal scaling of random walk
  {Metropolis} algorithms.}
\newblock {\em Annals of Applied Probability\/}, 7(1): 110--120.
\endbibitem

\bibitem[{Gelman and Hill(2007)}]{Gelman:2007}
Gelman, A. and Hill, J. (2007).
\newblock {\em Data Analysis Using Regression and Multilevel-Hierarchical
  Models\/}.
\newblock Cambridge University Press.
\endbibitem

\bibitem[{Gelman and Rubin(1992)}]{Gelman:1992}
Gelman, A. and Rubin, D.~B. (1992).
\newblock \enquote{Inference from iterative simulation using multiple sequences
  (with discussion).}
\newblock {\em Statistical Science\/}, 7: 457--511.
\endbibitem

\bibitem[{Gelman and Shirley(2011)}]{Gelman:2011}
Gelman, A. and Shirley, K. (2011).
\newblock \enquote{Inference from simulations and monitoring convergence.}
\newblock In {\em Handbook of Markov chain Monte Carlo\/}, chapter~6. CRC
  Press.
\endbibitem

\bibitem[{Geweke(1992)}]{Geweke:1992}
Geweke, J. (1992).
\newblock \enquote{Evaluating the accuracy of sampling-based approaches to the
  calculation of posterior moments.}
\newblock In {\em Bayesian Statistics 4\/}, 169--193. Oxford University Press.
\endbibitem

\bibitem[{Gilks et~al.(1994)Gilks, Roberts, and George}]{Gilks:1994}
Gilks, W.~R., Roberts, G.~O., and George, E.~I. (1994).
\newblock \enquote{Adaptive direction sampling.}
\newblock {\em Journal of the Royal Statistical Society: Series {D}\/}, 43(1):
  179--189.
\endbibitem

\bibitem[{Glynn and Rhee(2014)}]{Glynn:2014}
Glynn, P.~W. and Rhee, C.-H. (2014).
\newblock \enquote{Exact estimation for {Markov} chain equilibrium
  expectations.}
\newblock {\em Journal of Applied Probability\/}, 51: 377--389.
\endbibitem

\bibitem[{Heng and Jacob(2019)}]{Heng:2019}
Heng, J. and Jacob, P.~E. (2019).
\newblock \enquote{Unbiased {Hamiltonian Monte Carlo} with couplings.}
\newblock {\em Biometrika\/}, 106: 287 -- 302.
\endbibitem

\bibitem[{Hoffman and Sountsov(2022)}]{Hoffman:2022}
Hoffman, M. and Sountsov, P. (2022).
\newblock \enquote{Tuning-free generalized {Hamiltonian Monte Carlo}.}
\newblock {\em Artificial Intelligence and Statistics\/}, PMLR 151: 7799--7813.
\endbibitem

\bibitem[{Hoffman and Gelman(2014)}]{Hoffman:2014}
Hoffman, M.~D. and Gelman, A. (2014).
\newblock \enquote{The no-{U}-turn sampler: Adaptively setting path lengths in
  {Hamiltonian Monte Carlo}.}
\newblock {\em Journal of Machine Learning Research\/}, 15: 1593--1623.
\endbibitem

\bibitem[{Hoffman et~al.(2021)Hoffman, Radul, and Sountsov}]{Hoffman:2021}
Hoffman, M.~D., Radul, A., and Sountsov, P. (2021).
\newblock \enquote{An adaptive {MCMC} scheme for setting trajectory lengths in
  {Hamiltonian Monte Carlo}.}
\newblock {\em Artificial Intelligence and Statistics\/}, PMLR 130: 3907--3915.
\endbibitem

\bibitem[{Jacob et~al.(2020)Jacob, O'Leary, and Atchad\'e}]{Jacob:2020}
Jacob, P.~E., O'Leary, J., and Atchad\'e, Y.~F. (2020).
\newblock \enquote{Unbiased {Markov} chain {Monte Carlo} methods with
  couplings.}
\newblock {\em Journal of the Royal Statistical Society, Series B\/}, 82:
  543--600.
\endbibitem

\bibitem[{Jones et~al.(2006)Jones, Haran, Caffo, and Neath}]{Jones:2006}
Jones, G.~L., Haran, M., Caffo, B.~S., and Neath, R. (2006).
\newblock \enquote{Fixed-width output analysis for {Markov} chain {Monte}
  {Carlo}.}
\newblock {\em Journal of the American Statistical Association\/}, 101:
  1537--1547.
\endbibitem

\bibitem[{Lambert and Vehtari(2022)}]{Lambert:2022}
Lambert, B. and Vehtari, A. (2022).
\newblock \enquote{{$R^*$}: A robust {MCMC} convergence diagnostic with
  uncertainty using decision tree classifiers.}
\newblock {\em Bayesian Analysis\/}, 17: 353--379.
\endbibitem

\bibitem[{Lao et~al.(2020)Lao, Suter, Langmore, Chimisov, Saxena, Sountsov,
  Moore, Saurous, Hoffman, and Dillon}]{Lao:2020}
Lao, J., Suter, C., Langmore, I., Chimisov, C., Saxena, A., Sountsov, P.,
  Moore, D., Saurous, R.~A., Hoffman, M.~D., and Dillon, J.~V. (2020).
\newblock \enquote{tfp.mcmc: Modern {Markov} chain {Monte Carlo} tools built
  for modern hardware.}
\newblock {\em arXiv:2002.01184\/}.
\endbibitem

\bibitem[{Mackay(2003)}]{Mackay:2003}
Mackay, D.~J. (2003).
\newblock {\em Information Theory, Inference, and Learning Algorithms\/}.
\newblock Cambridge University Press.
\endbibitem

\bibitem[{Margossian and Gelman(2024)}]{Margossian:2023}
Margossian, C.~C. and Gelman, A. (2024).
\newblock \enquote{For how many iterations should we run Markov chain Monte
  Carlo?}
\newblock In {\em Handbook of Markov chain Monte Carlo\/}. CRC press,
  (upcoming) 2nd edition.
\endbibitem

\bibitem[{Margossian et~al.(2022)Margossian, Zhang, and
  Gillespie}]{Margossian:2022-torsten}
Margossian, C.~C., Zhang, Y., and Gillespie, W.~R. (2022).
\newblock \enquote{Flexible and efficient {Bayesian} pharmacometrics modeling
  using {Stan} and {Torsten}, part {I}.}
\newblock {\em CPT: Pharmacometrics \& Systems Pharmacology\/}, 11: 1151--1169.
\endbibitem

\bibitem[{Moins et~al.(2023)Moins, Arbel, Dutfoy, and Girard}]{Moins:2022}
Moins, T., Arbel, J., Dutfoy, A., and Girard, S. (2023).
\newblock \enquote{On the use of a local {$\widehat R$} to improve {MCMC}
  convergence diagnostic.}
\newblock {\em Bayesian Analysis\/}.
\endbibitem

\bibitem[{Neal(2001)}]{Neal:2001}
Neal, R.~M. (2001).
\newblock \enquote{Annealed importance sampling.}
\newblock {\em Statistics and Computing\/}, 11: 125--139.
\endbibitem

\bibitem[{Neal(2012)}]{Neal:2012}
--- (2012).
\newblock \enquote{{MCMC} using {H}amiltonian dynamics.}
\newblock In {\em Handbook of Markov Chain Monte Carlo\/}, chapter~5. CRC
  Press.
\endbibitem

\bibitem[{Nguyen et~al.(2022)Nguyen, Trippe, and Broderick}]{Nguyen:2022}
Nguyen, T.~D., Trippe, B.~L., and Broderick, T. (2022).
\newblock \enquote{Many processors, little time: {MCMC} for partitions via
  optimal transport couplings.}
\newblock {\em Artificial Intelligence and Statistics\/}, PMLR 151: 3483--3514.
\endbibitem

\bibitem[{Papaspiliopoulos et~al.(2007)Papaspiliopoulos, Roberts, and
  Sk\"old}]{Papaspiliopoulos:2007}
Papaspiliopoulos, O., Roberts, G.~O., and Sk\"old, M. (2007).
\newblock \enquote{A general framework for the parametrization of hierarchical
  models.}
\newblock {\em Statistical Science\/}, 22: 59--73.
\endbibitem

\bibitem[{Riabiz et~al.(2022)Riabiz, Chen, Cockayne, Swietach, Niederer,
  Mackey, and Oates}]{Riabiz:2022}
Riabiz, M., Chen, W., Cockayne, J., Swietach, P., Niederer, S.~A., Mackey, L.,
  and Oates, C.~J. (2022).
\newblock \enquote{Optimal thinning of MCMC output.}
\newblock {\em Journal of the Royal Statistical Society: Series B\/}, 84:
  1059--1081.
\endbibitem

\bibitem[{Riou-Durand et~al.(2023)Riou-Durand, Sountsov, Vogrinc, Margossian,
  and Power}]{Riou-Durand:2022b}
Riou-Durand, L., Sountsov, P., Vogrinc, J., Margossian, C.~C., and Power, S.
  (2023).
\newblock \enquote{Adaptive tuning for {Metropolis} adjusted {Langevin}
  trajectories.}
\newblock {\em Artificial Intelligence and Statistics\/}, PMLR 206: 8102 --
  8116.
\endbibitem

\bibitem[{Riou-Durand and Vogrinc(2022)}]{Riou-Durand:2022}
Riou-Durand, L. and Vogrinc, J. (2022).
\newblock \enquote{Metropolis adjusted {Langevin} trajectories: A robust
  alternative to {Hamiltonian Monte Carlo}.}
\newblock {\em arXiv:2202.13230\/}.
\endbibitem

\bibitem[{Robert and Casella(2004)}]{Robert:2004}
Robert, C.~P. and Casella, G. (2004).
\newblock {\em Monte Carlo Statistical Methods\/}.
\newblock Springer.
\endbibitem

\bibitem[{Roberts and Rosenthal(1998)}]{Roberts:1998}
Roberts, G.~O. and Rosenthal, J.~S. (1998).
\newblock \enquote{Optimal scaling of discrete approximations to {Langevin}
  diffusions.}
\newblock {\em Journal of the Royal Statistical Society, Series B\/}, 60:
  255--268.
\endbibitem

\bibitem[{Roberts and Rosenthal(2004)}]{Roberts:2004}
--- (2004).
\newblock \enquote{General state space {Markov chains and MCMC} algorithms.}
\newblock {\em Probability Surveys\/}, 1: 20 -- 71.
\endbibitem

\bibitem[{Rosenbrock(1960)}]{Rosenbrock:1960}
Rosenbrock, H.~H. (1960).
\newblock \enquote{An automatic method for finding the greatest or least value
  of a function.}
\newblock {\em Computer Journal\/}, 3: 175--184.
\endbibitem

\bibitem[{Rosenthal(2000)}]{Rosenthal:2000}
Rosenthal, J.~S. (2000).
\newblock \enquote{Parallel computing and {Monte Carlo} algorithms.}
\newblock {\em Far East Journal of Theoretical Statistics\/}, 4: 207--236.
\endbibitem

\bibitem[{Rubin(1981)}]{Rubin:1981}
Rubin, D.~B. (1981).
\newblock \enquote{Estimation in parallelized randomized experiments.}
\newblock {\em Journal of Educational Statistics\/}, 6: 377--400.
\endbibitem

\bibitem[{Sountsov and Hoffman(2021)}]{Sountsov:2021}
Sountsov, P. and Hoffman, M.~D. (2021).
\newblock \enquote{Focusing on difficult directions for learning {HMC}
  trajectory lengths.}
\newblock {\em arXiv:2110.11576\/}.
\endbibitem

\bibitem[{Sountsov et~al.(2020)Sountsov, Radul, and
  contributors}]{inferencegym:2020}
Sountsov, P., Radul, A., and contributors (2020).
\newblock \enquote{Inference Gym.}
\newline\urlprefix\url{https://pypi.org/project/inference_gym}
\endbibitem

\bibitem[{South et~al.(2021)South, Riabiz, Teymur, and Oates}]{South:2021}
South, L.~F., Riabiz, M., Teymur, O., and Oates, C.~J. (2021).
\newblock \enquote{Post-processing of {MCMC}.}
\newblock {\em Annual Review of Statistics and its Application\/}, 9: 1--30.
\endbibitem

\bibitem[{{TensorFlow Probability Development Team}(2023)}]{tfp:2023}
{TensorFlow Probability Development Team} (2023).
\newblock \enquote{TensorFlow Probability.}
\newline\urlprefix\url{https://www.tensorflow.org/probability}
\endbibitem

\bibitem[{Vats et~al.(2019)Vats, Flegal, and Jones}]{Vats:2019}
Vats, D., Flegal, J.~M., and Jones, G.~L. (2019).
\newblock \enquote{Multivariate output analysis for {Markov chain Monte
  Carlo}.}
\newblock {\em Biometrika\/}, 106: 321--337.
\endbibitem

\bibitem[{Vats and Knudson(2021)}]{Vats:2021}
Vats, D. and Knudson, D. (2021).
\newblock \enquote{Revisiting the {Gelman-Rubin} diagnostic.}
\newblock {\em Statistical Science\/}, 36: 518--529.
\endbibitem

\bibitem[{Vehtari(2022)}]{Vehtari:2022}
Vehtari, A. (2022).
\newblock \enquote{Bayesian workflow book - Digits.}
\newline\urlprefix\url{https://avehtari.github.io/casestudies/Digits/digits.html}
\endbibitem

\bibitem[{Vehtari et~al.(2021)Vehtari, Gelman, Simpson, Carpenter, and
  B\"urkner}]{Vehtari:2021}
Vehtari, A., Gelman, A., Simpson, D., Carpenter, B., and B\"urkner, P.-C.
  (2021).
\newblock \enquote{Rank-normalization, folding, and localization: An improved
  {$\widehat R$} for assessing convergence of {MCMC} (with discussion).}
\newblock {\em Bayesian Analysis\/}, 16: 667--718.
\endbibitem

\bibitem[{Wakefield(1996)}]{Wakefield:1996}
Wakefield, J. (1996).
\newblock \enquote{The {Bayesian} analysis of population pharmacokinetic
  models.}
\newblock {\em Journal of the American Statistical Association\/}, 91: 62--75.
\endbibitem

\bibitem[{Zhang et~al.(2020)Zhang, Gillespie, Bales, and Vehtari}]{Zhang:2020}
Zhang, Y., Gillespie, B., Bales, B., and Vehtari, A. (2020).
\newblock \enquote{Speed up population {Bayesian} inference by combining
  cross-chain warmup and within-chain parallelization.}
\newblock In {\em American Conference on Pharmacometrics\/}.
\endbibitem

\end{thebibliography}

\end{document}